\documentclass[final,5p,times,twocolumn]{elsarticle}

\usepackage{amssymb}
\usepackage{amsthm}
\usepackage{multirow}
\usepackage{wrapfig}
\usepackage{ulem}
\usepackage{tabularx}
\usepackage{booktabs}
\usepackage{stfloats}
\usepackage[utf8]{inputenc}
\usepackage{amsmath}
\usepackage{amsfonts}
\usepackage{epsfig}
\usepackage{psfrag}
\usepackage{pstricks}
\usepackage{algorithm}
\usepackage{graphicx}
\usepackage{url}
\newtheorem{theorem}{Theorem}

\newtheorem{proposition}[theorem]{Proposition}

\graphicspath{{./Figures/}}

\biboptions{comma,sort&compress}

\journal{Under Review}

\begin{document}

\begin{frontmatter}

\title{Generative and isoparametric geometric modeling of large-scale and multiscale microstructures}

\author[]{Guoyue Luo}
\author[]{Yuntao Ma}
\author[]{Qiang Zou\corref{cor}}\ead{qiangzou@cad.zju.edu.cn}

\cortext[cor]{Corresponding author.}
\address{State Key Laboratory of CAD$\&$CG, Zhejiang University, Hangzhou, 310058, China}

\begin{abstract}
As additive manufacturing advances toward higher printing resolution and larger build volumes, microstructures can be designed with finer geometric features over larger physical domains. This trend poses a fundamental challenge for geometric modeling: massive geometric details must be represented compactly, while their associations across scales must be maintained consistently. Existing methods cannot scale well to this requirement. Explicit representations suffer from prohibitive memory cost, and implicit representations remain compact only when microstructures admit analytic, periodic, or otherwise concise procedural descriptions. This paper proposes a new geometric modeling method that treats microstructure modeling as an on-demand generative process, rather than requiring the full instantiation of all geometric details. We first develop ExVCC, an extended volumetric Catmull–Clark spline representation that enables local spline refinement to go beyond tensor-product topology. Built on ExVCC, we introduce new shape-coding schemes and refinement rules that compactly encode large-scale geometric details and enable their localized evaluation through on-demand hierarchical refinement. To model geometric details across scales, we further propose an isoparametric representation in which details across scales are defined over a shared parametric domain using the same family of spline bases of ExVCC. This formulation turns the ExVCC's spline refinement hierarchy into a common framework for geometry encoding, on-demand generation, and cross-scale association, allowing geometric modifications to propagate automatically across scales. The effectiveness of the proposed method is demonstrated through a series of examples and comparisons.
\end{abstract}

\begin{keyword}
 Computer-Aided Design \sep Microstructure Modeling \sep On-demand Generation \sep Automatic Change Propagation \sep Volumetric Subdivision
\end{keyword}

\end{frontmatter}

\section{Introduction}
\label{sec:intro}
Additive manufacturing (AM) has advanced rapidly in recent years, making it practical to fabricate products with intricate interior microstructures~\cite{2020_lattice-review}. Such products are desirable in various applications, such as aerospace engineering, biomedical implants, and thermal management, because they can offer superior mechanical and functional properties, including lightweight performance, energy absorption, and multifunctional responses~\cite{2025_Multi-Physical-Lattice-Metamaterials,2021_property-tailoring}. These properties depend strongly on the geometric configuration of microstructures, including their topology, feature size, and spatial gradation. Geometric modeling thus plays a fundamental role in the design and manufacturing of microstructures~\cite{2025_Luo_review}.

As AM continues to improve in both printing resolution and build volume~\cite{2025_micro-nano-devices}, microstructures can be designed with finer geometric features over larger design domains. This trend poses new challenges for geometric modeling. As smaller features become fabricable, the number of geometric primitives in a fixed-volume part increases significantly. For example, the DARPA 2020 challenge~\cite{DARPA2020} reported that explicitly instantiating a $1\mathrm{m}^3$ design domain with interconnected $0.1\mathrm{mm}$ micro-trusses would require more than $100$~TB of memory. Moreover, higher resolution brings geometric details at multiple length scales into the same design domain, where they are coupled rather than existing as independent structures. Therefore, further development in geometric modeling is needed to compactly represent massive fine-scale geometric details while consistently maintaining their cross-scale associations.

Conventional geometric modeling methods, however, cannot scale well to these challenges. Explicit modeling methods, such as boundary representations (B-rep) and meshes~\cite{zou2019push,zou2023variational}, require full instantiation of all geometric details of interest, and therefore scale poorly to large-scale microstructure models~\cite{meta-meshing}. Implicit methods can be more compact when microstructures are described by concise functions or procedural rules~\cite{STL-free,nTopology,li2023xvoxel}. However, for highly irregular or locally varying microstructures, the underlying field often requires dense sampling, which can again lead to prohibitive memory costs. In addition, current methods often rely on post-processing or manual updates to maintain consistency across scales rather than automatic propagation, which is error-prone.

Recent research studies are seeking to improve scalability by avoiding full geometric instantiation. Compact representations based on programmed lattices, trans-similar structures, and procedural descriptions have been proposed to efficiently encode repeated or self-similar microstructures~\cite{2019_Programmed-Lattice,2022_Rossignac_Tran-Similar,2020_procedural_COTS}. However, their reliance on cell similarity or predefined rules restricts their applicability to general microstructures. Spline approaches have also been used to model cross-scale associations in microstructures~\cite{2023_Lin_freeform-multi-level,2018_elber_hierarchical,elber2023review}, but their reliance on tensor-product topology makes them less suitable for complex microstructures. Thus, effective geometric modeling of large-scale and multiscale microstructures remains an open problem.

To address this problem, this paper presents an generative~\footnote{Note that the generative here is different from that used in the AI field~\cite{zhang2025diffusion}.} and isoparametric modeling method for large-scale and multiscale microstructures. Instead of storing fully instantiated geometry, the method encodes a microstructure as a generative representation embedded in a volumetric parametric domain, so that fine geometry is generated only in queried regions. Besides, inspired by the isoparametric principle in IGA~\cite{hughes2005isogeometric}, we represent the design domain and microstructural geometries at different length scales over a shared family of hierarchical spline bases. As such, the continuity of spline bases helps maintain geometric continuity across adjacent microstructure regions, and their nestedness, expressed through the two-scale refinement relation~\cite{2016_spline-survey}, establishes cross-scale associations. Consequently, modifications can be propagated consistently both across adjacent regions and across scales, by construction rather than through post-processing. It should be noted that hierarchical splines are not new in microstructure modeling~\cite{2023_Lin_freeform-multi-level,elber2023review}, but they are mainly used for refinement or field approximation; by contrast, we primarily exploit the two-scale refinement relation to model cross-scale associations.

To make the proposed method applicable to complex design domains, we further develop an extended version of volumetric Catmull-Clark spline, termed ExVCC. Conventional hierarchical splines, such as the THB-splines used in existing microstructure modeling methods~\cite{2023_Lin_freeform-multi-level}, support local refinement but are limited to simple design domains with tensor-product topology. Volumetric Catmull-Clark splines, in contrast, can represent non-tensor-product volumetric domains~\cite{2020_Interpolatory-CC,2018_direct-limit-volumes}, but they do not directly provide the locally refinable hierarchical basis required by our generative modeling framework. ExVCC bridges this gap by extending volumetric Catmull-Clark splines into a locally refinable hierarchical form. As such, the subdivision in ExVCC enables on-demand generation of microstructure geometry, while the refinability establishes consistent cross-scale associations, even in volumetrically complex design domains.

The main contributions of this paper can be summarized as follows:
\begin{itemize}
    \item A new geometric modeling framework for large-scale and multiscale microstructures, enabling scalable on-demand geometry generation and automatic maintenance of geometric consistency across scales.
    \item A cross-scale associative modeling method based on the isoparametric principle, where spline continuity maintains same-scale geometric continuity and two-scale refinability establishes automatic associations across consecutive scales.
    \item The ExVCC method, an extended volumetric Catmull--Clark spline method that supports hierarchical refinement over non-tensor-product volumetric domains.
\end{itemize}

The remainder of this paper is organized as follows. Section~\ref{sec:related} reviews related work. Sections~\ref {sec:framework} and~\ref{sec:implementation} present the modeling framework and implementations, respectively. Section~\ref{sec:results} presents examples and comparisons, followed by conclusions in Section~\ref{sec:conclusion}.

 \section{Related Work}
\label{sec:related}

\subsection{Geometric Modeling of Large-Scale Microstructures}
\label{sec:related_large_scale}

Large-scale microstructure modeling often induces an explosion of geometric data if explicit methods like B-rep or meshes are used~\cite{soap-film-inspired,quasi-optimal-shape-design,TPMS2STEP}, challenging both representation and downstream processing. Prior work has tackled this issue from two complementary perspectives: how microstructures are represented, and how they are evaluated for downstream tasks. 

On the representation side, various compact encodings have been proposed to avoid explicit storage of massive geometric detail. Program-based representations~\cite{2019_Programmed-Lattice} can compactly represent large-scale microstructures, but they typically assume periodic arrangements. Function-based representations (e.g., TPMS)~\cite{2022_efficient-TPMS,2023_Implicit-conforming,1995_Pasko_F-rep} model geometry implicitly with low memory consumption,  yet their expressiveness is limited, as many microstructures cannot be faithfully captured by such functional forms.

On the evaluation side, scalable pipelines have been developed for handling large geometric data. Out-of-core approaches~\cite{meta-meshing,2021_memory-efficient} reduce peak memory via chunk-based processing, but still require full instantiation of detailed geometry, leaving total storage and computational costs largely unchanged. By contrast, on-demand evaluation~\cite{2022_Rossignac_Tran-Similar,2020_procedural_COTS} performs local, query-driven evaluation without global instantiation, reducing peak memory and avoiding unnecessary computation on irrelevant regions. However, it relies on structured assumptions (e.g., cell similarity) for fast localization, which may limit design space.

\subsection{Geometric Modeling of Multiscale Microstructures}

Multiscale microstructures offer increased flexibility for property tailoring by incorporating geometric details at multiple scales. The geometric modeling of multiscale microstructures has been explored through two main paradigms: explicit refinement-based and implicit refinement-based approaches.

Explicit refinement-based methods construct microstructures by placing microstructure cells directly in the macroshape’s physical domain and introducing multiscale detail through explicit geometric/topological refinement of these cells, e.g., cell insertion~\cite{2024_fractal_honeycomb,2016_vertex_hierarchical_honeycombs} or cell splitting~\cite{2016_rhombic_infill,2018_adaptive-quadtree,2023_TO_self-supporting,2022_Load-Adaptive-Design_cell-split}. They are effective for initial microstructure generation but typically provide limited support for consistent multiscale editing, because cross-scale correspondences are not intrinsically encoded in the representation. For instance, modifications to the macroshape often do not automatically propagate to the instantiated microstructure geometry.

In contrast, implicit refinement-based approaches define microstructures functionally, with multiscale characteristics introduced via hierarchical basis refinement~\cite{2023_Lin_freeform-multi-level} or function composition~\cite{2013_Yoo_hierarchical-TPMS,2021_TPMS-boolean_perodic,2018_elber_hierarchical,2019_elber_Micro-Tiles,Persistent-homology-driven-TPMS}. While effective for multiscale design, consistent cross-scale modification propagation is typically achieved only when a nested refinement structure is available (e.g., THB-splines), so that coarse-to-fine relationships are explicitly defined. Moreover, existing nested-refinable formulations~\cite{2023_Lin_freeform-multi-level} for microstructure modeling rely on tensor-product parametric domains, which limit their applicability to geometrically complex volumetric shapes with nontrivial topology. In contrast, the proposed ExVCC splines in this work accommodate control meshes with extraordinary vertices, allowing them to handle complex volumetric topologies and greatly expanding the design space of microstructures.

\section{The Proposed Modeling Framework}
\label{sec:framework}

To address the modeling challenge posed by large-scale microstructures, we use an on-demand generation strategy, in which geometric details are generated only when and where they are required. This strategy is implemented through hierarchical splines, which organize the control mesh into a coarse-to-fine hierarchy and allow refinement to be applied only to selected cells (see Fig.~\ref{fig:hierarchical}). We begin with a coarse control mesh, which provides a compact coarse representation over the entire design domain. When additional detail is required, we refine only the selected cells, forming a nested hierarchy of progressively finer cells while leaving the rest of the domain unchanged. The fine-scale geometry is then generated on demand, only within this refined region. To handle volumetrically complex shapes that cannot be easily parameterized by tensor-product splines, we realize our framework using volumetric Catmull–Clark subdivision splines~\cite{2018_direct-limit-volumes}. This representation readily accommodates nontrivial topology, making it well-suited for microstructure modeling in complex volumetric domains.

With hierarchical splines, we also realize multiscale associations that enable automatic change propagation across scales and maintain connectivity among adjacent cells. To achieve multiscale modeling, we represent the macro-shape and the microstructure by the same spline basis over the same parametric domain. This shared parameterization establishes a direct association between macro-scale deformation and micro-scale evaluation, so the microstructure follows the macro-shape and remains conformal under shape changes. Across scales, the microstructure at each refinement level is built from the same spline basis family, and the hierarchy records parent–child refinement relations. These relations provide the associations needed to propagate geometric modifications throughout the hierarchy. At a certain scale, connectivity across neighboring cells is enforced at the representation level, because both sides evaluate the shared interface with the same basis functions and coefficients, ensuring the boundary geometry coincides exactly.

\begin{figure}[t]
\centering
\includegraphics[width=0.43\textwidth]{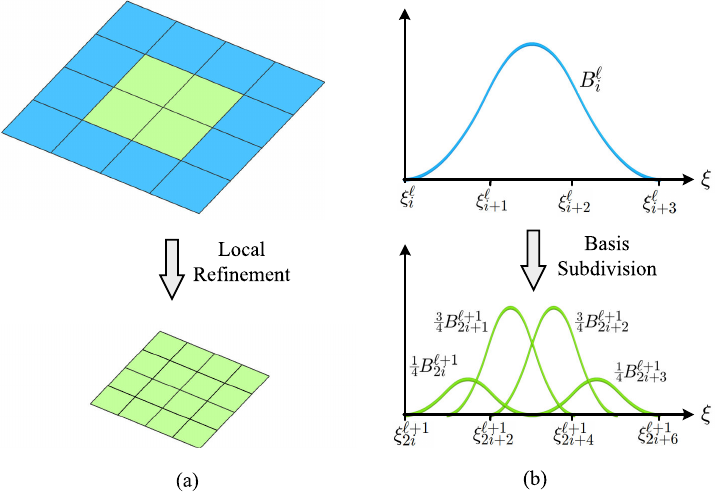}
\caption{The illustration of hierarchical spline representation: (a) the local refinement; and (b) the linear refinability (within knot span $[\xi_i^\ell,\xi_{i+3}^\ell]$, a quadratic B-spline basis $B_i^\ell=\frac{1}{4}B_{2i}^{\ell+1}+\frac{3}{4}B_{2i+1}^{\ell+1}+\frac{3}{4}B_{2i+2}^{\ell+1}+\frac{1}{4}B_{2i+3}^{\ell+1}$).}
\label{fig:hierarchical}
\end{figure}

\begin{figure}[t]
\centering
\includegraphics[width=0.48\textwidth]{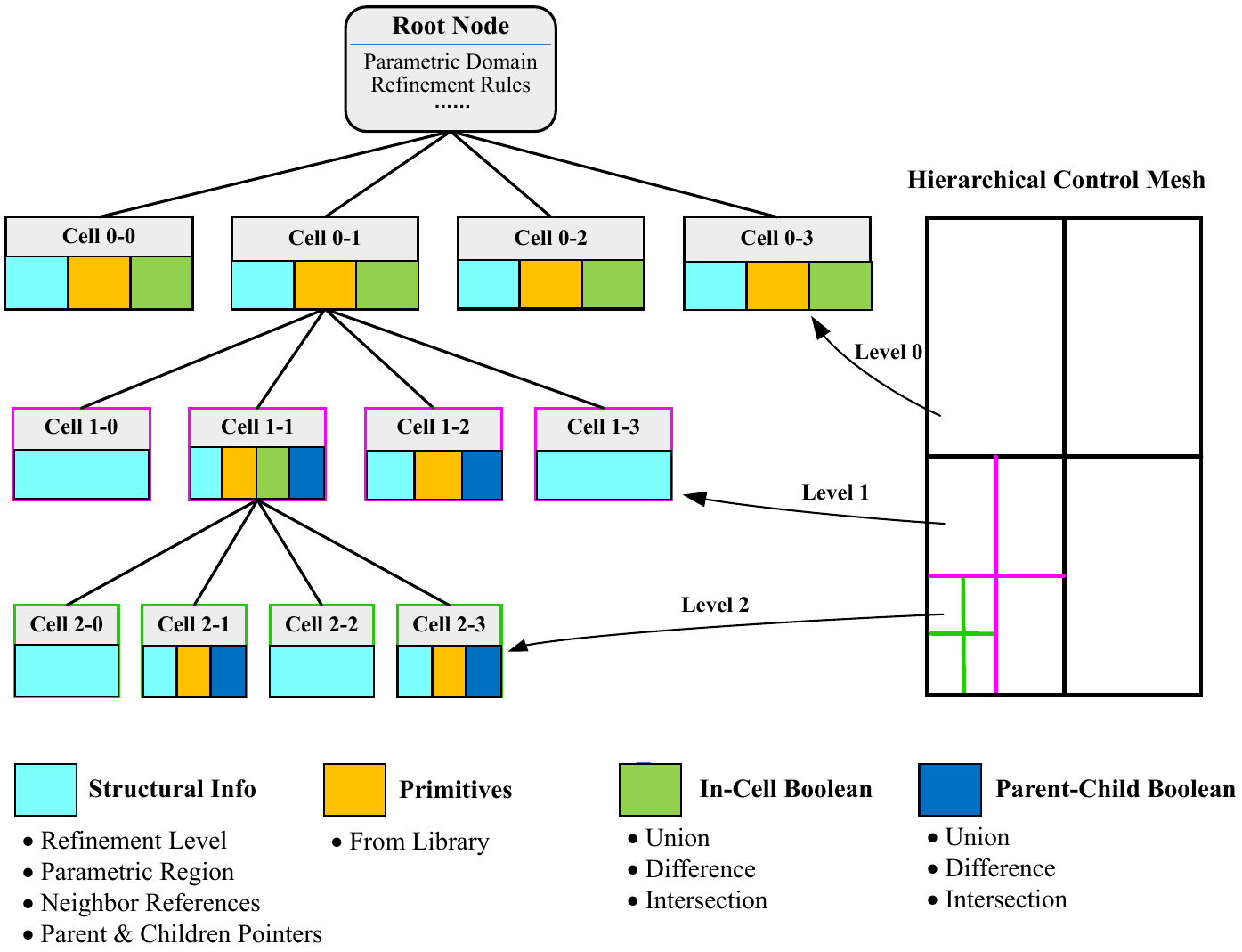}
\caption{The illustration of tree-structured representation over hierarchical control mesh.}
\label{fig:tree}
\end{figure}

Built upon the above ideas, we organize the modeling workflow as a tree-structured representation over the hierarchical control mesh (see Fig.~\ref{fig:tree}). The tree is anchored by a virtual root node that stores global context (e.g., the parametric domain and the refinement rules). Its children correspond to the cells of the initial coarse control mesh. Each subsequent node is associated with exactly one cell at a specific refinement level. When a cell is refined, the node associated with that cell becomes an internal node and spawns child nodes that correspond to the refined subcells whose parametric regions partition the parent cell. Each node stores the structural information, including: (1) the refinement level; (2) the parametric region covered by the cell; (3) pointers to its parent and children; and (4) neighbor references to face-sharing cells at the same level. The node also carries the modeling content, including: (1) microstructure primitives; (2) composition operations within cells; and (3) composition operations across parent–child nodes. This tree structure primarily stores the rules and necessary information, rather than the final geometry itself. This makes it an inherently scalable representation.

With this tree structure, microstructure evaluation can be formulated as a recursive computation on the tree. With a refinement depth given, we traverse the hierarchy to collect the active leaves at that depth. Each node stores a node-local implicit field definition given by its primitive expression and an associated Boolean operation that combines this local field with the inherited field from its parent. The field of a child is obtained by Boolean-combining the parent field with the child’s node-local field, and the field on an active leaf is computed by applying this recursion along the unique root-to-leaf path. We then extract the corresponding isosurface within each active leaf cell and map it to physical space under the shared macro parameterization, which enables the geometric association between microstructure and macro-shape.

The same root-to-leaf recursion used for evaluation makes edit propagation automatic. When a coarse node is modified, all leaves in its subtree are updated by re-evaluating the recursion under the revised node definition. When a fine-level node is modified, the update stays confined to the corresponding local subtree, since only leaves whose paths pass through the edited node need to be recomputed. In addition, connectivity is maintained by construction at a fixed scale because neighboring leaves evaluate their shared interface using the same spline basis functions and coefficients. The implicit field therefore agrees on both sides of the interface, and the extracted boundary geometry matches exactly across cell boundaries.

We next introduce the primitive library and the Boolean composition operations used to build node-local fields within a cell and to combine them across refinement levels.

\textbf{Microstructure primitives.}
Microstructure primitives are the basic building blocks from which the overall microstructure is constructed. Our library includes several representative primitives (see Table~\ref{tab:primitive_overview}), such as cylinders, spheres, and ellipsoids. Each primitive is represented as an implicit scalar field over a cell’s parametric domain, expressed by basis functions and coefficients. Within the same cell, all primitives share the same spline basis and differ only in their coefficient vectors.

\begin{table*}[t]
\centering
\caption{Typical primitives supported by our library.}
\label{tab:primitive_overview}
\begin{tabularx}{0.83\textwidth}{c c c c}
\hline
No. & Primitive & Definition & Model \\
\hline
\#1 & Plate & Knots, Basis Functions, Coefficients & \raisebox{-0.5\height}{\centering \includegraphics[width=1.2cm]{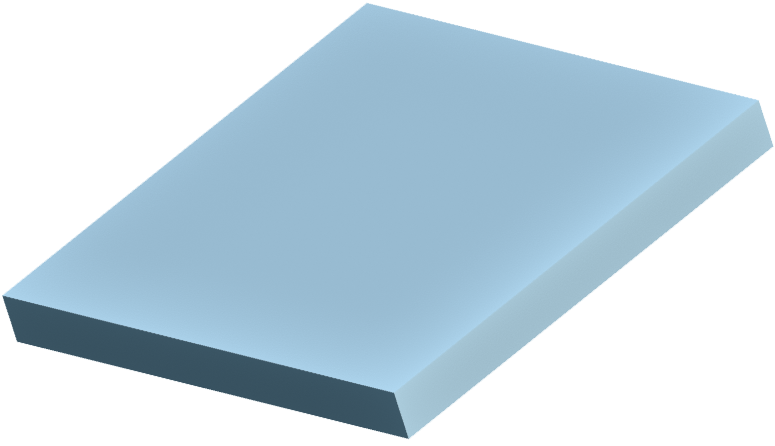}} \\
\#2 & Cylinder & Knots, Basis Functions, Coefficients & \raisebox{-0.5\height}{\centering \includegraphics[width=1.cm]{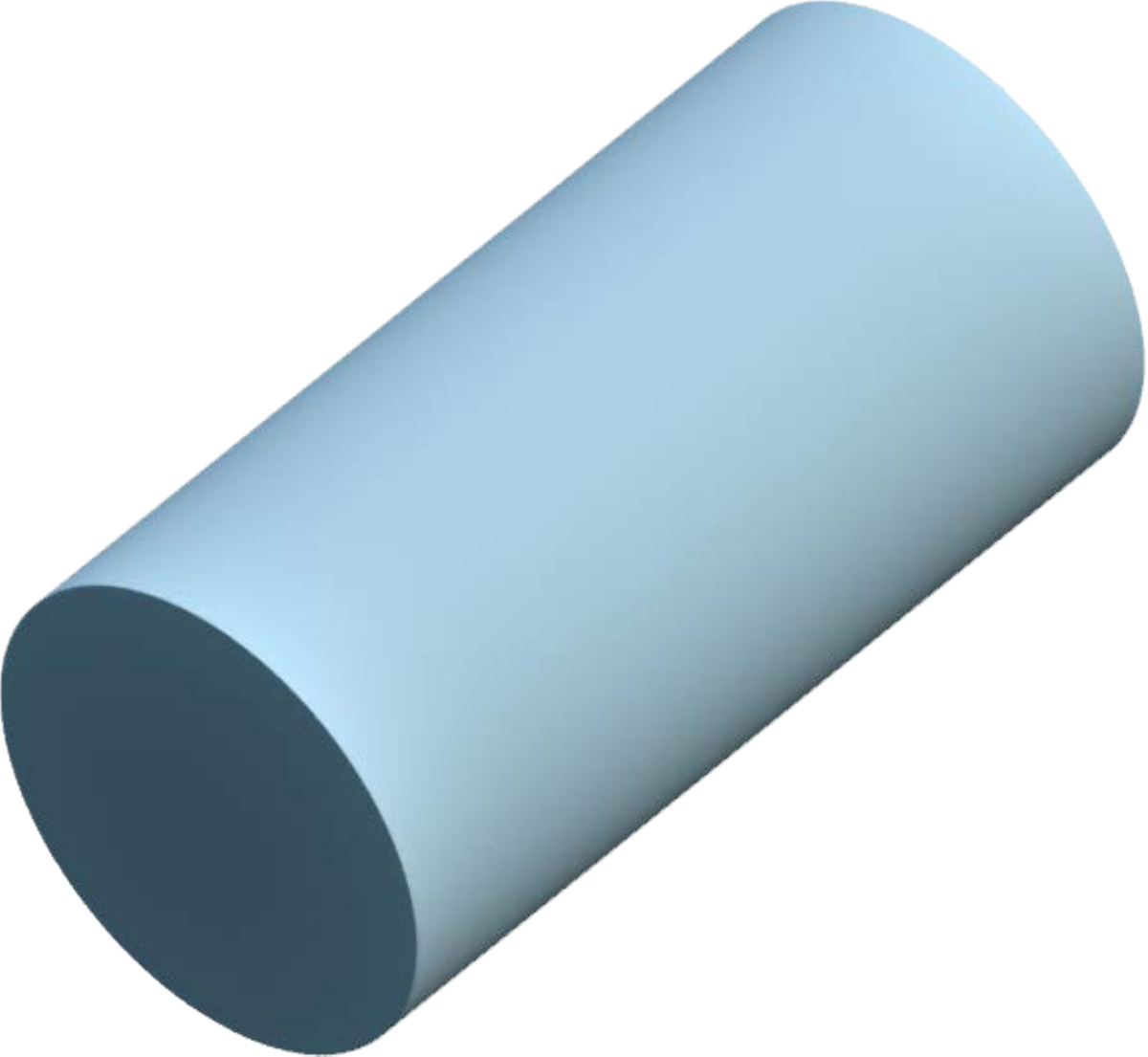}} \\
\#3 & Sphere & Knots, Basis Functions, Coefficients & \raisebox{-0.5\height}{\centering \includegraphics[width=1.cm]{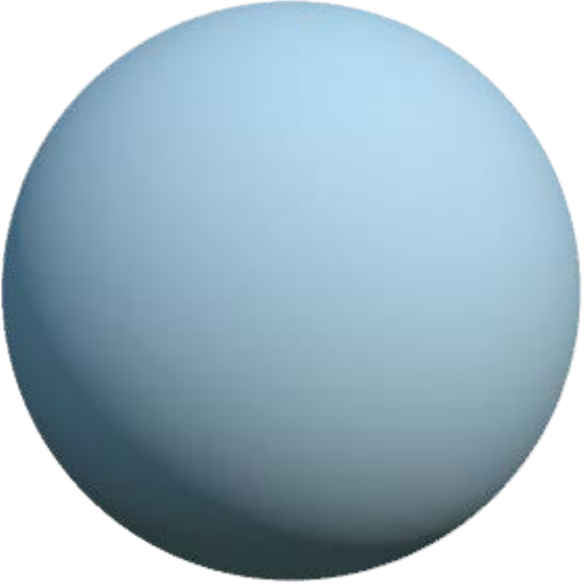}} \\
\#4 & Ellipsoid & Knots, Basis Functions, Coefficients & \raisebox{-0.5\height}{\centering \includegraphics[width=1.cm]{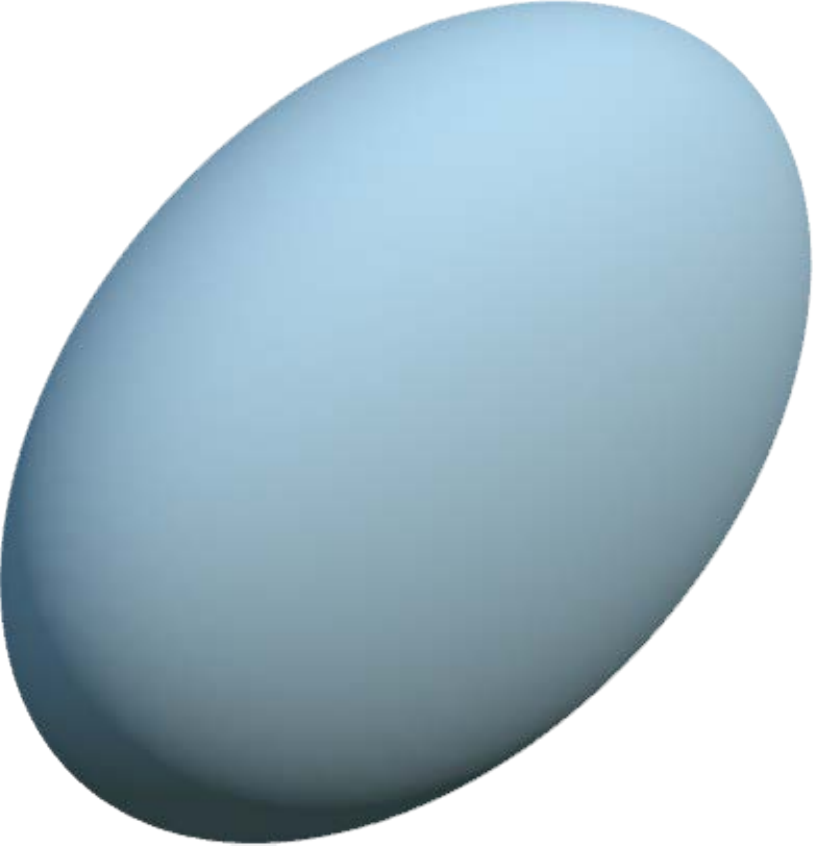}}  \\
\#5 & Shell Unit 1 & Plate, Transformation, Boolean Operation &\raisebox{-0.5\height}{\centering \includegraphics[width=1.2cm]{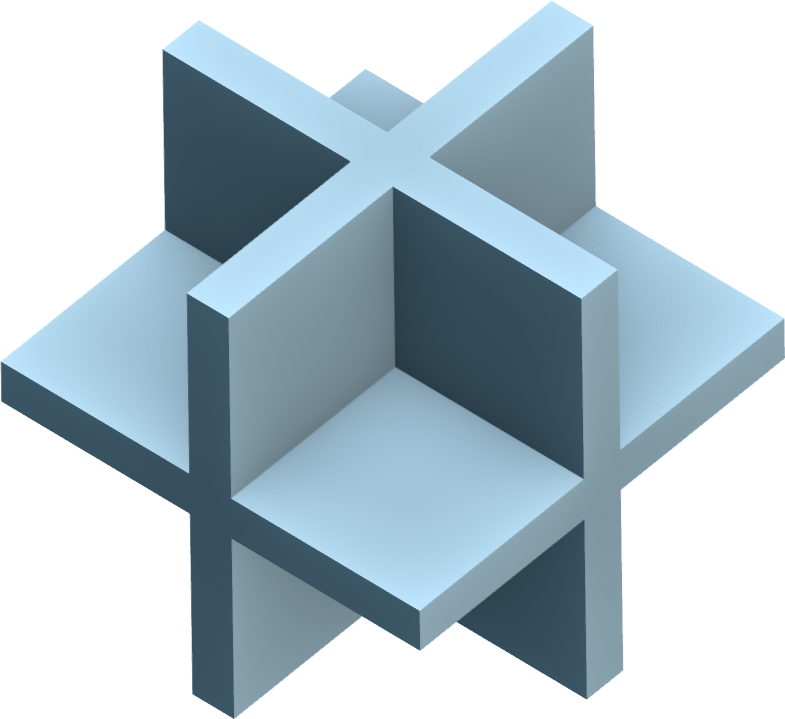}}  \\
\#6 & Shell Unit 2 & Plate, Deformation, Transformation, Boolean Operation& \raisebox{-0.5\height}{\centering \includegraphics[width=1.cm]{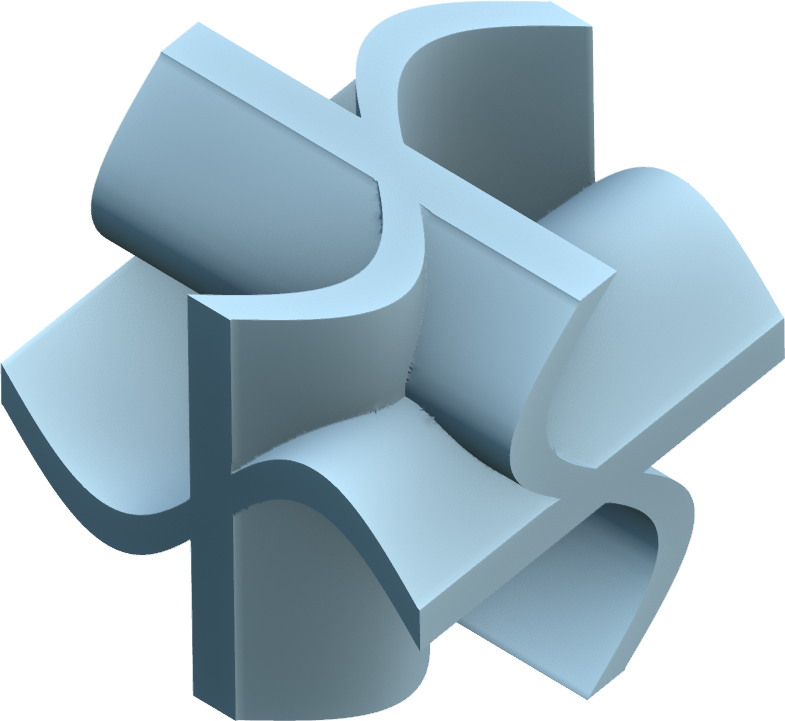}} \\
\#7 & Shell Unit 3 & Sphere, Strut Unit, Transformation, Boolean Operation& \raisebox{-0.5\height}{\centering \includegraphics[width=0.94cm]{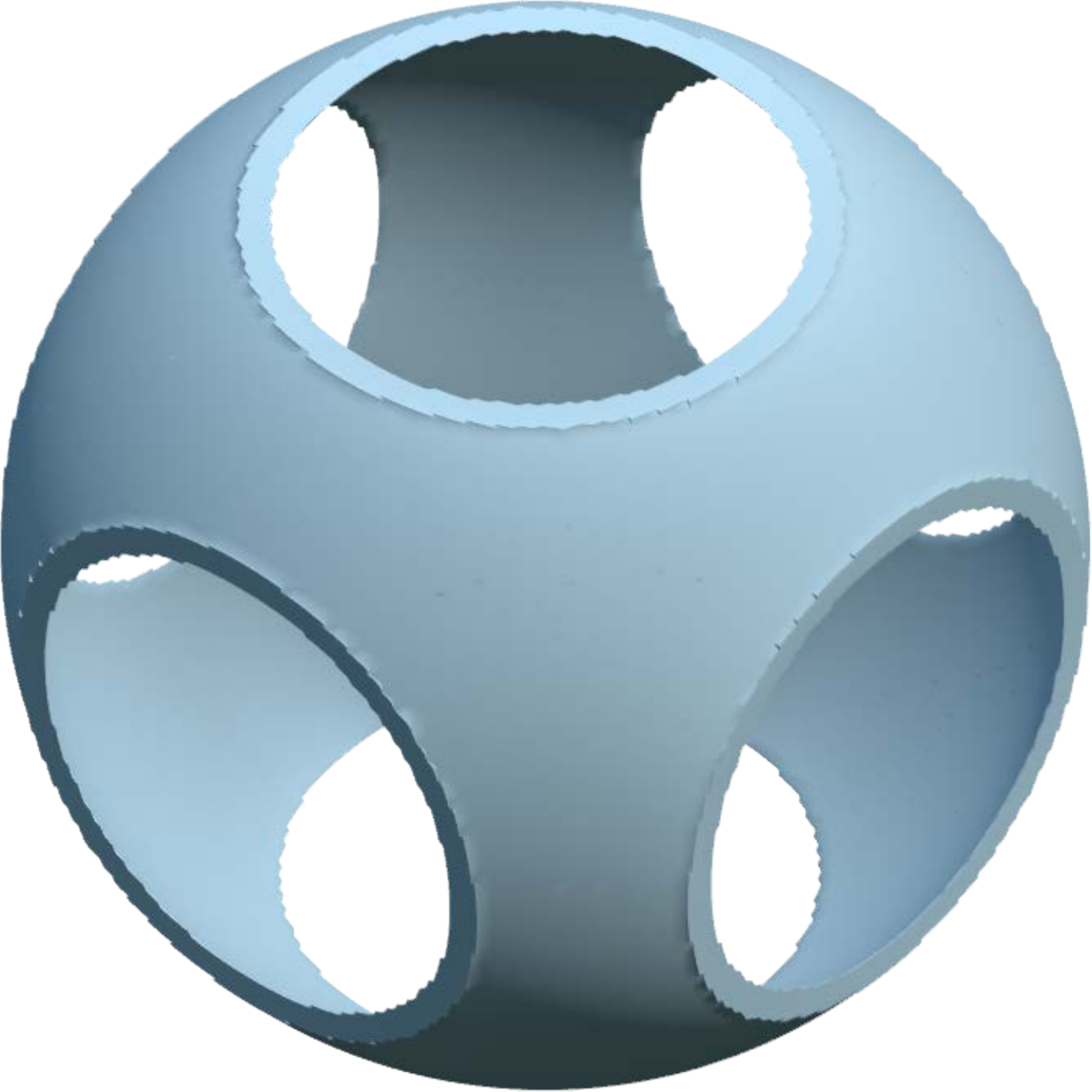}} \\
\#8 & Strut Unit 1 & Cylinder, Transformation, Boolean Operation &\raisebox{-0.5\height}{\centering \includegraphics[width=1.cm]{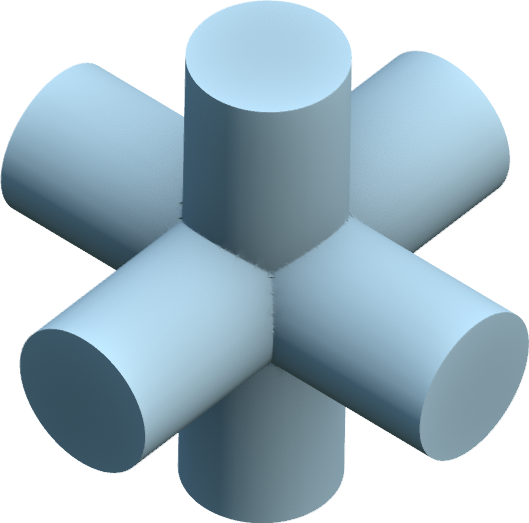}}  \\
\#9 & Strut Unit 2 & Cylinder, Deformation, Transformation, Boolean Operation &\raisebox{-0.5\height}{\centering \includegraphics[width=1.1cm]{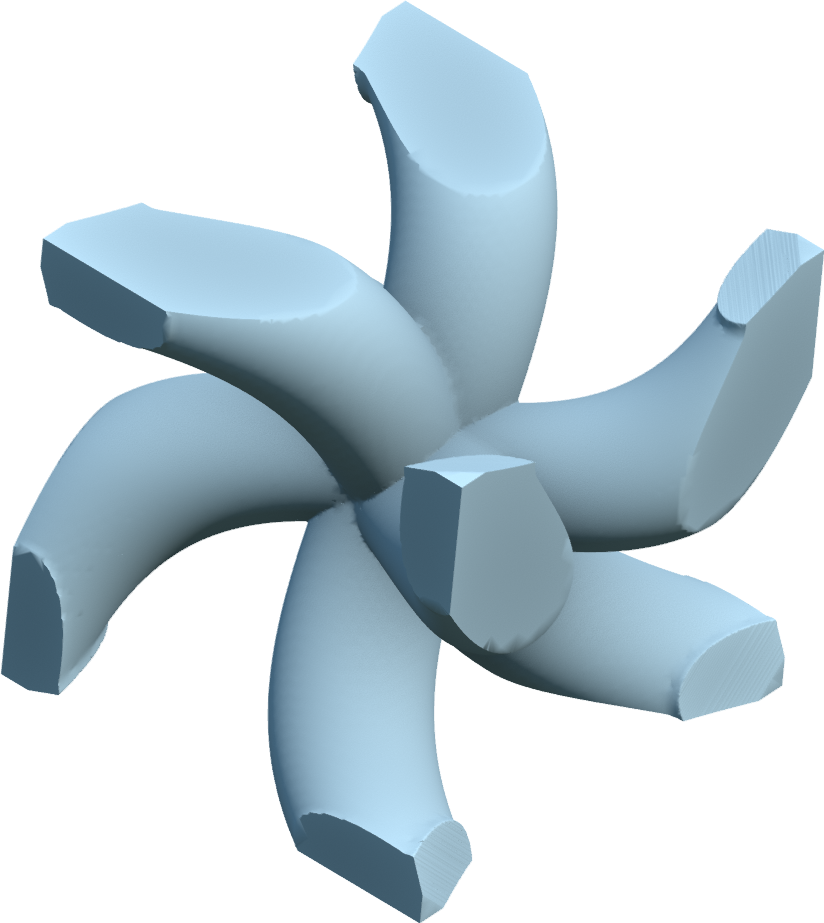}}  \\
\#10 & Solid Unit 1 & Cube, Cylinder, Transformation, Boolean Operation &\raisebox{-0.5\height}{\centering \includegraphics[width=1.cm]{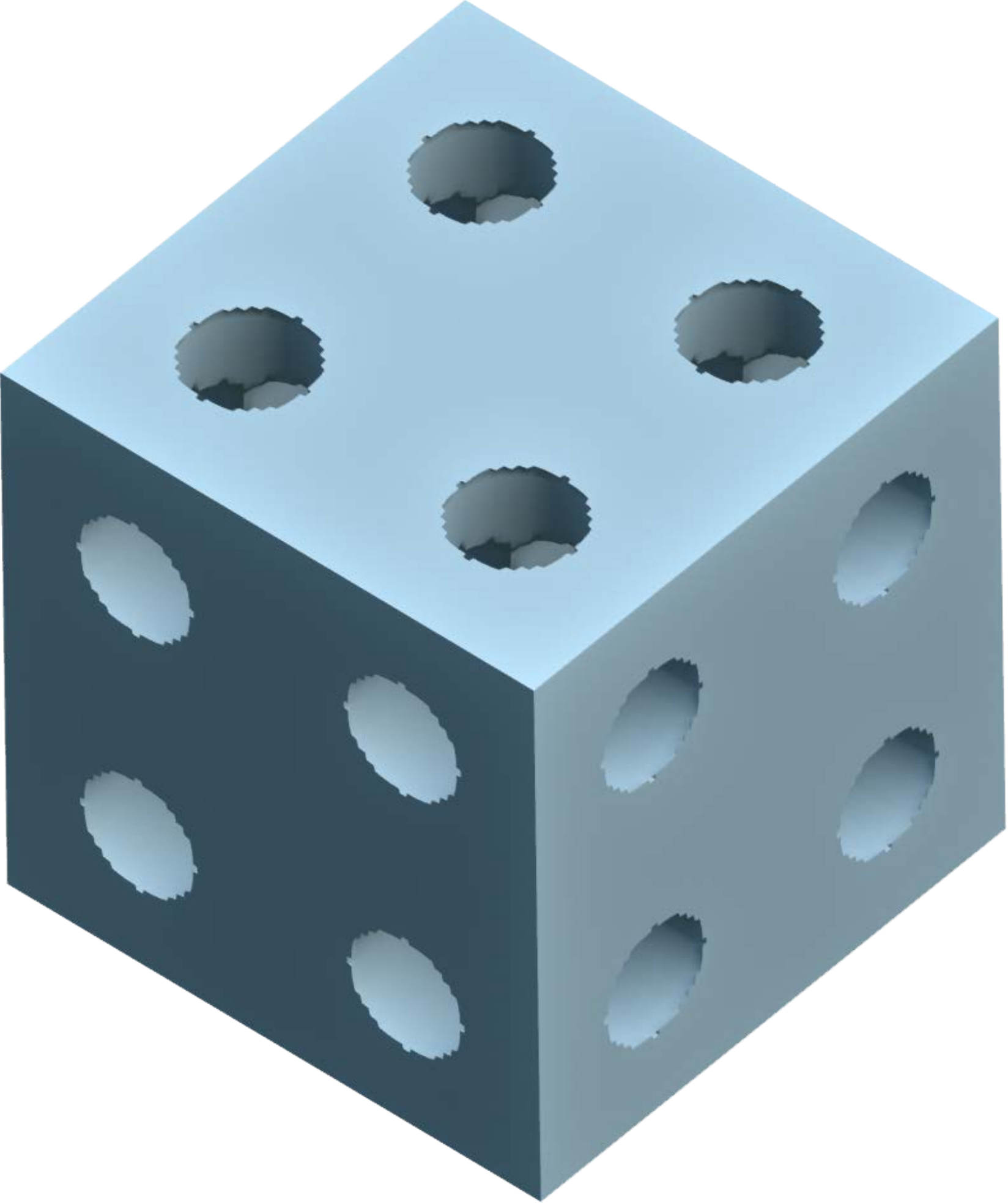}}  \\
\#11 & Solid Unit 2 & Cube, Ellipsoid, Cylinder, Sphere, Transformation, Boolean Operation &\raisebox{-0.5\height}{\centering \includegraphics[width=1.cm]{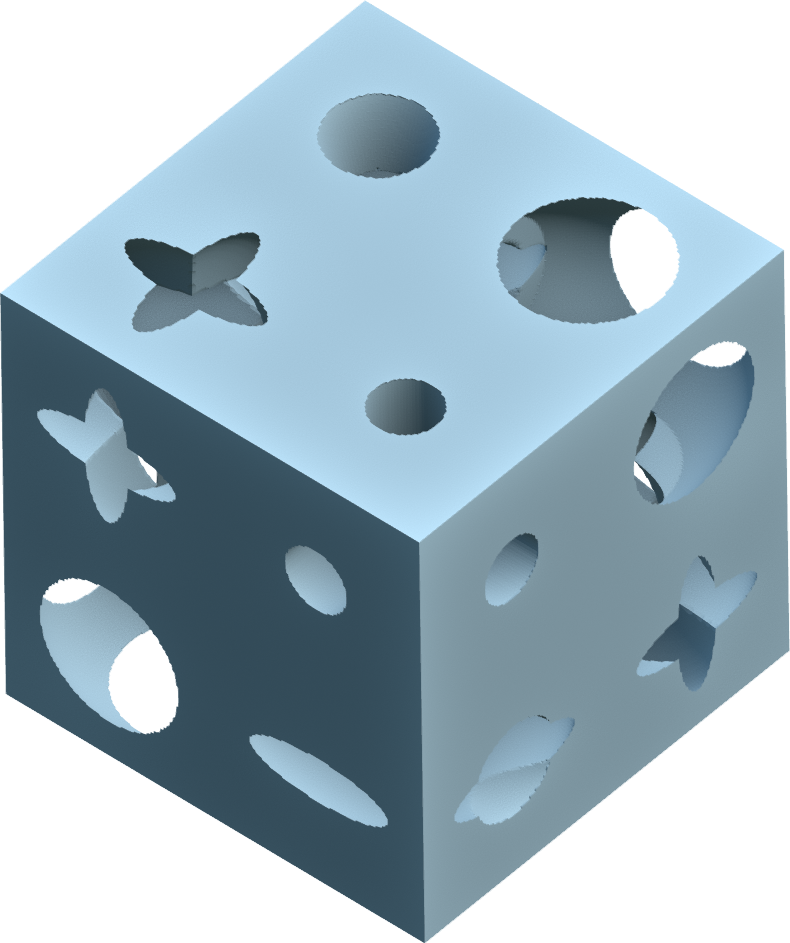}}  \\
\hline
\end{tabularx}
\end{table*}

\textbf{Primitive combinations.}
Two classes of primitive combinations are supported. The first operates within a cell, where multiple primitives defined on the same cell are composed through Boolean operations (union, intersection, and difference) to form a more complex single-scale microstructure.

The second class operates across refinement levels, where primitives associated with a parent cell are combined with those defined in its child cells. This cross-scale composition enables multiscale primitives and is realized through three typical operations:
\begin{itemize}
    \item \textbf{Geometry-preserving refinement:} increasing representation resolution and degree of freedom without changing geometry, and the additional degrees of freedom enable the introduction of freeform details by adjusting the corresponding coefficients;
    \item \textbf{Copy-and-scale variation:} replicating parent-level patterns into child cells with natural scaling;
    \item \textbf{Cross-scale Boolean composition:} combining parent- and child-level fields so that fine-scale features can add to or carve from coarse-scale geometry.
\end{itemize}

It should be noted that both the primitive library and the composition operations are designed to be extensible. While the examples in this paper focus on a representative subset, additional primitives and richer field composition operators can be incorporated within the same framework to further broaden the design space for multiscale microstructures.

\section{Implementation via Volumetric Subdivision Splines }
\label{sec:implementation}
This section presents the implementation of our modeling framework using volumetric Catmull--Clark (VCC) subdivision splines~\cite{2018_direct-limit-volumes}. To support generative and iso-parametric modeling, the VCC basis is extended to a locally refinable hierarchical form, namely ExVCC. Based on this hierarchical spline representation, multiscale microstructures are constructed and evaluated within the proposed framework.

\subsection{Introduction to VCC}
VCC subdivision is a well-established geometric modeling method, and here we provide a brief summary of it. VCC subdivision defines a refinement operator on a volumetric control mesh, producing a nested sequence of refined meshes that converges to a smooth volumetric limit geometry~\cite{2020_Interpolatory-CC}. At each subdivision step, a coarse cell is split (in the parametric sense) into eight child sub-cells, yielding a finer volumetric layout together with an updated set of control vertices governed by the Catmull--Clark subdivision rules.

This refinement process induces a spline representation by associating each cell with a local parametric domain (an element) and expressing the limit geometry over that element as a linear combination of control vertices weighted by VCC basis functions~\cite{2018_direct-limit-volumes}. For regular elements (Fig.~\ref{fig:element-type}(a)), the basis reduces to standard tricubic B-spline basis functions. For irregular elements containing extraordinary vertices (EVs; valence $\neq 6$) and/or extraordinary edges (EEs; valence $\neq 4$) (Fig.~\ref{fig:element-type}(b)), the basis functions are defined by the subdivision construction, analogous to Stam’s framework~\cite{1998_Stam_evaluation}. The local neighborhood is subdivided until the query point lies in a non-irregular subelement, and then the basis is expressed as~\cite{2018_direct-limit-volumes}:
\begin{equation}
\mathbf{B}(u,v,w)
=
(\mathbf{V}^{-1})^{T}
\mathbf{\Lambda}^{n-1}
\left(\mathbf{P}_{k}\bar{\mathbf{A}}\mathbf{V}\right)^{T}
\mathbf{N}(u,v,w)
\end{equation}
where $\bar{\mathbf{A}}$ is the extended subdivision matrix, $\mathbf{P}_k$ is the volumetric picking matrix ($k=0,\ldots,6$), and $(\mathbf{\Lambda},\mathbf{V})$ is the eigenstructure of the subdivision operator (Please refer to~\cite{2018_direct-limit-volumes} for the construction procedures of these matrices). The volumetric basis $\mathbf{N}(u,v,w)$ is constructed as a tensor product of the 2D Catmull-Clark basis and a regular cubic B-spline basis. It should be noted that a local subdivision step in the neighborhood of an EV produces eight children that are not uniformly regular; as shown in Fig.~\ref{fig:local_subdiv}, four become regular away from EEs ($k\in\{0,1,2,4\}$), three form layered configurations associated with EEs ($k\in\{3,5,6\}$), and one remains irregular ($k=7$).

VCC subdivision splines can represent both macro-shapes and microstructures in a unified volumetric spline form. However, the direct application of VCC to the multiscale modeling could cause inconsistencies at the coarse-fine interfaces because it lacks native support for local refinement, and to the best of our knowledge, no publicly available work has reported local refinement for VCC. Therefore, we develop a locally refinable hierarchical extension of VCC (ExVCC) to model multiscale microstructures. Unlike conventional THB-splines, which are constrained by a tensor-product topology, ExVCC is defined over control meshes with extraordinary vertices, thereby enabling the representation of complex volumetric topologies and significantly expanding the design space for multiscale structures.

\begin{figure}[t]
\centering
\includegraphics[width=0.49\textwidth]{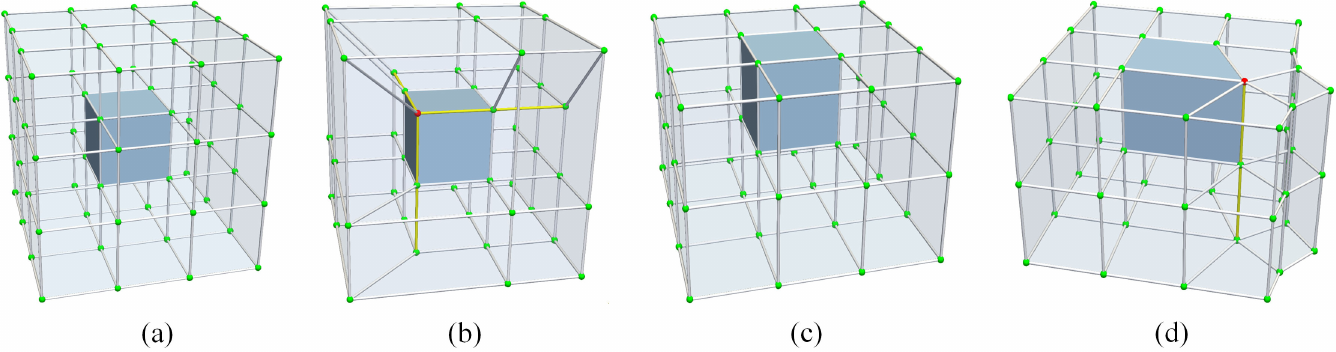}
\caption{Elements categorized by local configurations (EVs in red; EEs in yellow): (a)(c) regular elements and (b)(d) irregular elements.}
\label{fig:element-type}
\end{figure}

\begin{figure}[t]
\centering
\includegraphics[width=0.38\textwidth]{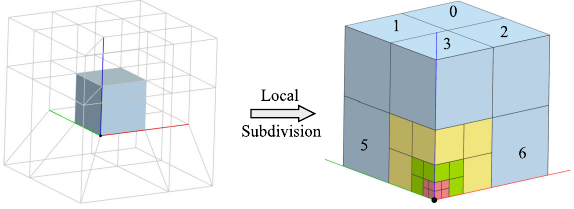}
\caption{The local subdivision of an irregular element.}
\label{fig:local_subdiv}
\end{figure}

\subsection{Extension of VCC}
\label{sec:VTHCC}
The ExVCC basis construction recursively follows three steps: (1) determining active level-$(\ell+1)$ functions over a refined domain; (2) truncating the affected level-$\ell$ functions with respect to that domain; and (3) collecting all resulting functions to form the hierarchical basis.

\begin{figure*}[t]
\centering
\includegraphics[width=0.8\textwidth]{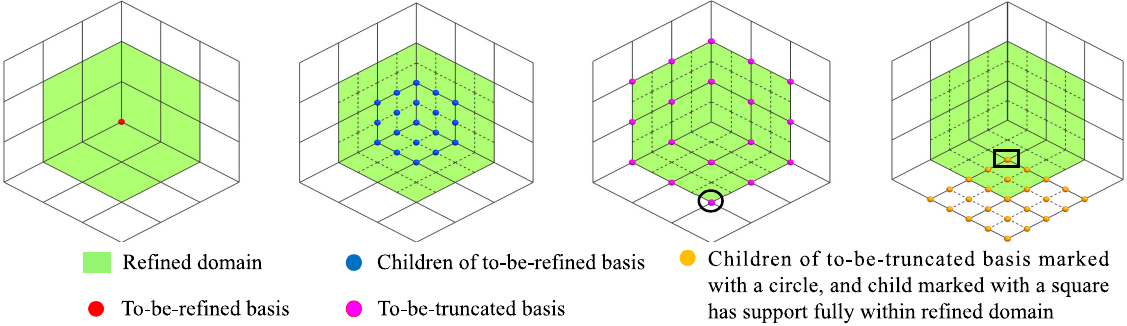}
\caption{The illustration of the construction of ExVCC basis functions.}
\label{fig:THCC_construct}
\end{figure*}

\textbf{Determination of active level-$(\ell+1)$ functions.}
Consider two consecutive levels, level $\ell$ with basis functions $\mathbf{B}^{\ell}$ and level $(\ell+1)$.
We start by specifying the to-be-refined basis functions $\mathbf{B}^\ell_r$ and level-$\ell$ domain $\Omega^{\ell}$ as the union of supports of level-$\ell$ basis functions $\mathbf{B}^\ell_r$ (see Fig.~\ref{fig:THCC_construct}),
denoted by 
\begin{equation}
\Omega^{\ell} := \bigcup_{B_i^{\ell}\in \mathbf{B}^{\ell}_{r}}\operatorname{supp}(B_i^{\ell})
\end{equation}

In the tricubic setting, each basis function has support spanning its two-ring neighborhood elements. $\Omega^{\ell}$ is typically composed of several local domains $\{\Omega^{\ell}_k\}$.
Each local domain $\Omega^{\ell}_k$ is then refined by subdividing into eight subdomains. These subdomains constitute the parametric domains of the corresponding level-$(\ell+1)$ volumetric elements. The refined domain at level $\ell+1$ is then defined as the union of all subdomains generated from the refinement,
which we denote by $\Omega^{\ell+1}$ and is identical to $\Omega^{\ell}$.

Once $\Omega^{\ell+1}$ is determined, the level-$(\ell+1)$ basis functions of the sub-patches corresponding to the subdomains can be derived as in~\cite{2018_direct-limit-volumes}. Recall that via the Catmull–Clark subdivision matrix, the control points at level $\ell$ and level $\ell+1$ are related as
\begin{equation}
    P^{\ell+1}_j=\sum^{n^\ell-1}_{i=0}c^{\ell}_{ji}P^\ell_i
\end{equation}
where $c^{\ell}_{ji}$ are the subdivision (refinement) coefficients and $n^\ell$ is the number of control points at level $\ell$.

We say that $B^{\ell+1}_j$ is a child of $B^{\ell}_i$ if $c^{\ell}_{ji}\neq 0$.
In particular, refining $\Omega^{\ell}$ activates all children generated from the marked set $\mathbf{B}^{\ell}_{r}$, and these children have supports fully
contained in the refined domain $\Omega^{\ell+1}$ by construction.
Accordingly, we define the set of active level-$(\ell+1)$ basis functions by support inclusion:
\begin{equation}
\mathbf{B}^{\ell+1}_{a}
=
\left\{
B^{\ell+1}_i \in \mathbf{B}^{\ell+1}
\;:\;
\operatorname{supp}\!\left(B^{\ell+1}_i\right)\subseteq \Omega^{\ell+1}
\right\}
\label{eq:active_fine_vthcc}
\end{equation}

\textbf{Truncation of affected level-$\ell$ functions.}
To identify the basis functions to be truncated, we start by considering the active basis functions at level $\ell$, denoted by $\mathbf{B}^{\ell}_a$. We define the set of to-be-truncated functions $\mathbf{B}^{\ell}_t$ as the subset of active functions in $\mathbf{B}^{\ell}_a$ whose children have support fully contained within the refined domain $\Omega^{\ell+1}$. Specifically, for each $B^{\ell}_i \in \mathbf{B}^{\ell}_a$, we check whether any of its children intersect with the children of the marked set $\mathbf{B}^{\ell}_r$. If so, $B^{\ell}_i$ is identified for truncation, and we write:

\begin{equation}
\mathbf{B}^{\ell}_t
=
\left\{
B^{\ell}_i \in \mathbf{B}^{\ell}_a
\;:\;
\text{chd}(B^{\ell}_i) \cap \text{chd}(\mathbf{B}^{\ell}_r) \neq \emptyset
\right\}
\label{eq:active_fine_vthcc_truncate}
\end{equation}

For each basis function $B^\ell_i \in \mathbf{B}^\ell_t$, the two-scale refinement relation expresses it as a linear combination of its level-$(\ell+1)$ descendants,
\begin{equation}
B_i^{\ell} = \sum_{j} c_{ji}^\ell\, B_j^{\ell+1}
\end{equation}
where $c_{ji}^\ell$ are refinement coefficients from the Catmull--Clark subdivision matrix; the proof of this linear relation is given in \textbf{Proposition~\ref{proposition}} (at the last of Section~\ref{sec:VTHCC}).
We then perform truncation by discarding the contributions of those children $B_j^{\ell+1}$ whose supports are entirely contained within the refined domain $\Omega^{\ell+1}$. This truncation process can be expressed as:
\begin{equation}
\text{trun}(B_i^{\ell}) = \sum_{\substack{\operatorname{supp}(B_j^{\ell+1}) \nsubseteq \Omega^{\ell+1}}} c_{ji}^\ell\, B_j^{\ell+1}
\end{equation}

\textbf{Collection of basis functions.}
Finally, the ExVCC at level $\ell+1$ is obtained by collecting:
(i) the newly activated fine-level functions, (ii) the truncated versions of the affected coarse-level functions, and (iii) the coarse-level functions that do not contribute to $\Omega^{\ell+1}$:
\begin{equation}
\mathbf{B}^{\ell+1}_{TH}
=
\mathbf{B}^{\ell+1}_{a}
\ \cup\
\operatorname{trun}_{\ell+1}\big(\mathbf{B}^{\ell}_t\big)
\ \cup\
\left(\mathbf{B}^{\ell}_a \setminus \mathbf{B}^{\ell}_t\right)
\label{eq:collect_vthcc}
\end{equation}
where $\mathbf{B}^{\ell+1}_{a}$ denotes the active fine-level basis functions, $\operatorname{trun}_{\ell+1}(\mathbf{B}^{\ell}_t)$ represents the truncated coarse-level functions, and $\mathbf{B}^{\ell}_a \setminus \mathbf{B}^{\ell}_t$ corresponds to the coarse-level functions that are not truncated.

\begin{proposition}
\label{proposition}
Volumetric Catmull--Clark basis functions are refinable: for any level $\ell$, coarse-level basis functions can be
expressed exactly as a linear combination of basis functions at level $\ell+1$, with coefficients given by the Catmull––Clark subdivision matrix.
\end{proposition}

\begin{proof}
The mixed regular/layered/irregular outcome of one-step VCC subdivision motivates a patch-wise refinability analysis. Let $\mathbf{B}^{\ell}$ denote the vector of volumetric Catmull--Clark basis functions associated with an irregular element $\Omega^{\ell}$ at level $\ell$. Using the explicit direct-limit formulation~\cite{2018_direct-limit-volumes} for volumetric Catmull--Clark solids, we have

\begin{equation}
\mathbf{B}^{\ell}
=
(\mathbf{V}^{-1})^{T}\,
\Lambda^{\,n-\ell-1}\,
(\mathbf{P}_{k}\bar{\mathbf{A}}\mathbf{V})^{T}\,
\mathbf{N}
\label{eq:1}
\end{equation}

For the $k_{th}$ child sub-element $\Omega^{\ell+1}_k$ at level $\ell+1$, we denote its basis vector by $\mathbf{B}^{\ell+1}_k$:
\begin{equation}
\mathbf{B}^{\ell+1}_k :=
\begin{cases}
\mathbf{N}, & k\in\{0,1,2,3,4,5,6\}\\
\mathbf{B}^{\ell+1}, & k=7
\end{cases}
\label{eq:child_basis_def}
\end{equation}

We consider the eight child sub-elements $\Omega^{\ell+1}_k$ produced by one subdivision step and distinguish two cases depending on the child type.

For the child that remains irregular ($k=7$), denote its level-$(\ell+1)$ basis vector by $\mathbf{B}^{\ell+1}$. By the same direct-limit formulation,

\begin{equation}
\begin{aligned}
\mathbf{B}^{\ell+1}
&= (\mathbf{V}^{-1})^{T}\,
\mathbf{\Lambda}^{\,n-(\ell+1)-1}\,
(\mathbf{P}_{k}\bar{\mathbf{A}}\mathbf{V})^{T}\,
\mathbf{N} \\
&= (\mathbf{\Lambda}\mathbf{V}^T)^{-1}\,
\mathbf{\Lambda}^{\,n-\ell-1}\,
(\mathbf{P}_{k}\bar{\mathbf{A}}\mathbf{V})^{T}\,
\mathbf{N} \\
\end{aligned}
\label{eq:2}
\end{equation}
Since $\mathbf{A}\mathbf{V}=\mathbf{V}\mathbf{\Lambda}$ and $\mathbf{\Lambda}=\mathbf{\Lambda}^T$, it follows that

\begin{equation}
    \mathbf{\Lambda}\mathbf{V}^T=\mathbf{V}^T\mathbf{A}^T
    \label{eq:trans}
\end{equation}
By plugging Eq.~(\ref{eq:trans}) into Eq.~(\ref{eq:2}), we obtain
\begin{equation}
\begin{aligned}
\mathbf{B}^{\ell+1}
&= (\mathbf{A}^{T})^{-1}\,
(\mathbf{V}^{-1})^{T}\,
\Lambda^{\,n-\ell-1}\,
(\mathbf{P}_{k}\bar{\mathbf{A}}\mathbf{V})^{T}\,
\mathbf{N}\\
&= (\mathbf{A}^{T})^{-1}\,
\mathbf{B}^{\ell}
\end{aligned}
\label{eq:3}
\end{equation}
Therefore, the refinability relation along the irregular branch is
\begin{equation}
    \mathbf{B}^{\ell}=\mathbf{A}^T\mathbf{B}^{\ell+1}
\end{equation}
where $\mathbf{A}$ is the subdivision matrix, obtained as a square form of the extended subdivision matrix $\bar{\mathbf{A}}$ (see~\cite{2018_direct-limit-volumes} for further details), and $(\mathbf{\Lambda}, \mathbf{V})$ denotes the eigenstructure of $\mathbf{A}$.

For the remaining children ($k\in\{0,1,2,3,4,5,6\}$), the local element type becomes regular or layered after one subdivision step. We analyze the refinability by taking advantage of the refinement invariance of the geometry over each child sub-element.

Let $\mathbf{C}^{\ell}$ denote the control vertices at level $\ell$. One subdivision step produces the child-level control vertices by

\begin{equation}
    \mathbf{C}^{\ell+1}=\bar{\mathbf{A}}\mathbf{C}^{\ell}
\end{equation}
and the control vertices associated with the child sub-element $k$ are obtained by picking

\begin{equation}
    \mathbf{C}_k^{\ell+1}=\mathbf{P}_k\mathbf{C}^{\ell+1}
\end{equation}

On a regular (resp. layered) child sub-element, the geometry can be evaluated using the corresponding level-$\ell+1$ basis vector $\mathbf{B}^{\ell+1}_k$ (equivalent to $\mathbf{N}$), hence

\begin{equation}
\begin{aligned}
\mathbf{s}\big|_{\Omega^{\ell+1}_k}
&=
\big(\mathbf{C}^{\ell+1}\big)^{T}\mathbf{P}_{k}^{T}\mathbf{B}^{\ell+1}_k\\
&=(\mathbf{C}^{\ell})^T\bar{\mathbf{A}}^T\mathbf{P}_{k}^{T}\mathbf{B}^{\ell+1}_k
\end{aligned}
\label{eq:skn_dual}
\end{equation}

On the other hand, by refinement invariance, the same restricted geometry can also be written in terms of the parent basis restricted to $\Omega^{\ell+1}_k$

\begin{equation}
\begin{aligned}
\mathbf{s}\big|_{\Omega^{\ell+1}_k}
&=
\big(\mathbf{C}^{\ell}\big)^{T}\mathbf{B}^{\ell}
\end{aligned}
\end{equation}

As a result, the parent basis restricted to the child sub-element admits the linear relation
\begin{equation}
\mathbf{B}^{\ell}\big|_{\Omega^{\ell+1}_k}
=
(\mathbf{P}_k \bar{\mathbf{A}})^{T}\,\mathbf{B}^{\ell+1}_k,
\quad k\in\{0,1,2,3,4,5,6\}
\end{equation}    
\end{proof}

\subsection{Iso-Parametric Modeling of Microstructures}
In our method, the microstructure is formed as an assembly of microstructure primitives. This section, therefore, concentrates on the spline representation of primitives and on multiscale primitive design enabled by hierarchical refinement, as well as how the resulting implicit fields are evaluated to obtain the final microstructure geometry.

\subsubsection{Single-Scale Primitive Construction}
\label{sec:single-scale}
As stated in Section~\ref{sec:framework}, microstructure geometry is implicitly defined by composing a finite set of implicit primitives through a Boolean composition graph. Accordingly, the construction here focuses on (i) the implicit spline representation of each primitive, and (ii) the implicit Boolean composition that yields a single composite field.

Let $\mathbf{B}=\{B_i\}$ denote the ExVCC spline basis active on the current local parametric domain. Each primitive $p$ is represented as a scalar spline field:
\begin{equation}
\phi_p(\mathbf{u})=\sum_{B_i\in\mathbf{B}} \alpha_{p,i}\, B_i(\mathbf{u}),
 \quad\mathbf{u}\in[0,1]^3
\label{eq:primitive-spline}
\end{equation}
where $\alpha_{p,i}$ are the coefficients of basis functions. Within the same local parametric domain, primitives share the identical basis and differ only in their coefficients. For these typical primitives—such as cylinders, spheres, and ellipsoids—we provide a method for computing the coefficients in~\ref{app:coeff}.

We also provide two types of operations for primitives: rigid transformations and Boolean compositions.
For a primitive field $\phi_p(\mathbf{u})$, a rigid transform $T_p$ is expressed as:
\begin{equation}
\tilde{\phi}_p(\mathbf{u})=\phi_p\!\left(T_p^{-1}(\mathbf{u})\right)
\label{eq:primitive-transform}
\end{equation}

Given two implicit fields $\phi_A(\mathbf{u})$ and $\phi_B(\mathbf{u})$, the Boolean composition between them is conducted using the standard max/min construction for union, intersection, and difference:
\begin{equation}
\label{eq:boolean}
\begin{aligned}
\phi_{A\cup B}(\mathbf{u}) &= \max\!\big(\phi_A(\mathbf{u}),\,\phi_B(\mathbf{u})\big)\\
\phi_{A\cap B}(\mathbf{u}) &= \min\!\big(\phi_A(\mathbf{u}),\,\phi_B(\mathbf{u})\big)\\
\phi_{A\setminus B}(\mathbf{u}) &= \min\!\big(\phi_A(\mathbf{u}),\, -\phi_B(\mathbf{u})\big)
\end{aligned}
\end{equation}

This implicit spline representation of primitives serves as the basic building block for subsequent multiscale primitive construction and on-demand microstructure evaluation.

\subsubsection{Multiscale Primitive Construction}
\label{sec:multiscale}

We support three multiscale operations for spline-based implicit primitives, as described in Section~\ref{sec:framework}. 
We start from a primitive $p$ defined at level $\ell$ in the active spline basis $\mathbf{B}^{\ell}=\{B_i^{\ell}\}$:
\begin{equation}
\phi_p^{\ell}(\mathbf{u})=\sum_i \alpha_{p,i}^{\ell}\,B_i^{\ell}(\mathbf{u})
\end{equation}
When a region is refined to level $\ell+1$, the representation space changes to the finer basis $\mathbf{B}^{\ell+1}=\{B_j^{\ell+1}\}$. 
To keep the geometry unchanged while gaining more degrees of freedom for subsequent edits, we rewrite the same primitive field using the refined basis.

\textbf{Geometry-preserving refinement operation.}
The ExVCC basis satisfies the two-scale relation:
\begin{equation}
B_i^{\ell}(\mathbf{u})=\sum_{j} c_{ji}^{\ell}\,B_j^{\ell+1}(\mathbf{u})
\label{eq:ms-two-scale}
\end{equation}
where $c_{ji}^{\ell}$ are one-step refinement coefficients. Substituting Eq.~\eqref{eq:ms-two-scale} into the level-$\ell$ expansion yields the coefficients on level $\ell+1$:
\begin{equation}
\alpha_{p,j}^{\ell\rightarrow \ell+1}=\sum_i c_{ji}^{\ell}\,\alpha_{p,i}^{\ell}
\label{eq:ms-prolong-coef}
\end{equation}
Using these coefficients, we define
\begin{equation}
\phi_{p}^{\ell\rightarrow \ell+1}(\mathbf{u})
= \sum_{j}\alpha_{p,j}^{\ell\rightarrow \ell+1}\,B_j^{\ell+1}(\mathbf{u})
\label{eq:ms-prolong}
\end{equation}
which is the same primitive field as $\phi_p^{\ell}$, but expressed in the refined spline basis $\mathbf{B}^{\ell+1}$. 
This refinement step preserves the primitive geometry, while providing a finer representation with additional degrees of freedom on level $\ell+1$. 

It should be noted that these additional degrees of freedom allow for the introduction of freeform details into the microstructure, which can be controlled by adjusting the corresponding coefficients. The values of these coefficients can either be explicitly specified by the user or computed through a physics-driven optimization process.

\textbf{Copy-and-scale operation.}
To replicate parent-level patterns inside refined child cells, we instantiate additional level-$(\ell+1)$ primitives by copying the parent primitive coefficients.
Let $\boldsymbol{\alpha}_{p}^{\ell}=\{\alpha_{p,i}^{\ell}\}$ denote the coefficient vector of primitive $p$ at level $\ell$, stored in a fixed local ordering.
For each refined child cell $c$, we create a primitive instance $p_c\in\mathcal{P}_{\mathrm{fine}}^{\ell+1}$ by assigning the same coefficient values to the basis functions active in cell $c$ (using the same local ordering):
\begin{equation}
\alpha_{p_c,j}^{\ell+1} \;=\; \alpha_{p,j}^{\ell}
\label{eq:ms-copy}
\end{equation}

The resulting primitive field is then represented on the level-$(\ell+1)$ basis as:
\begin{equation}
\phi_{p_c}^{\ell+1}(\mathbf{u})
=\sum_{j}\alpha_{p_c,j}^{\ell+1}\,B_j^{\ell+1}(\mathbf{u})
\label{eq:ms-copy-field}
\end{equation}

It should be noted that no explicit scaling transformation is applied here. The copied primitive becomes smaller automatically because the level-$(\ell+1)$ basis functions have smaller supports over smaller refined cells.

\begin{figure}[t]
\centering
\includegraphics[width=0.42\textwidth]{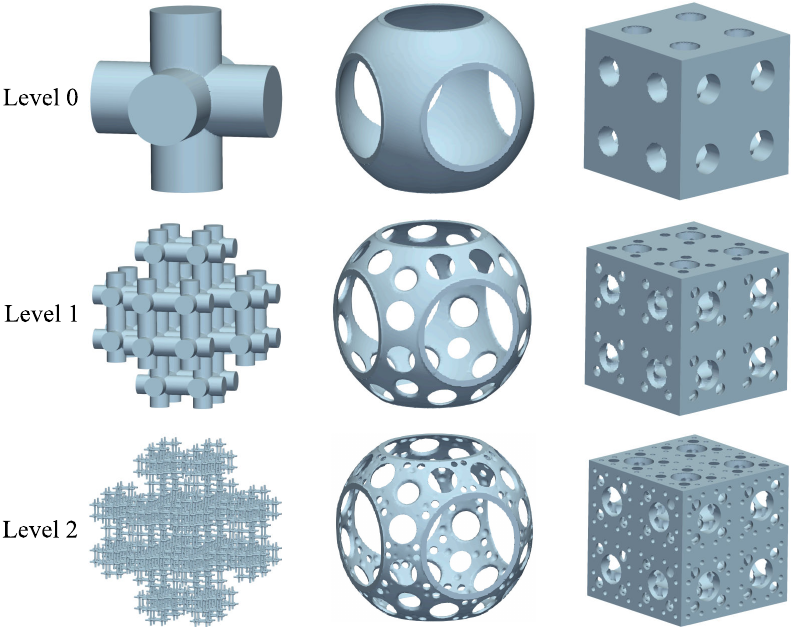}
\caption{Three typical microstructure primitives at three consecutive hierarchical levels.}
\label{fig:multiscale_cell}
\end{figure}

\begin{table*}[b]
\centering
\caption{Summary statistics of the test models in Case Study 1.}
\label{tab:model_statistics}
\begin{tabular*}{0.65\textwidth}{cccccc}
\toprule
Models & No. of CVs & No. of HEs & No. of EVs & No. of EEs & No. of Genus \\
\midrule
 Torus & 96 & 54 & 0 & 0 & 1 \\
Bitorus & 212 & 126 & 4 & 16 & 2 \\
Plane & 288 & 148 & 12 & 36 & 0 \\
Cow & 400 & 240 & 0 & 0 & 0 \\
Wrench & 360 & 176 & 16 & 48 & 1 \\ 
Bracket & 1092 & 604 & 56 & 168 & 6 \\
\bottomrule
\end{tabular*}
\end{table*}

\begin{figure*}[t]
\centering
\includegraphics[width=0.9\textwidth]{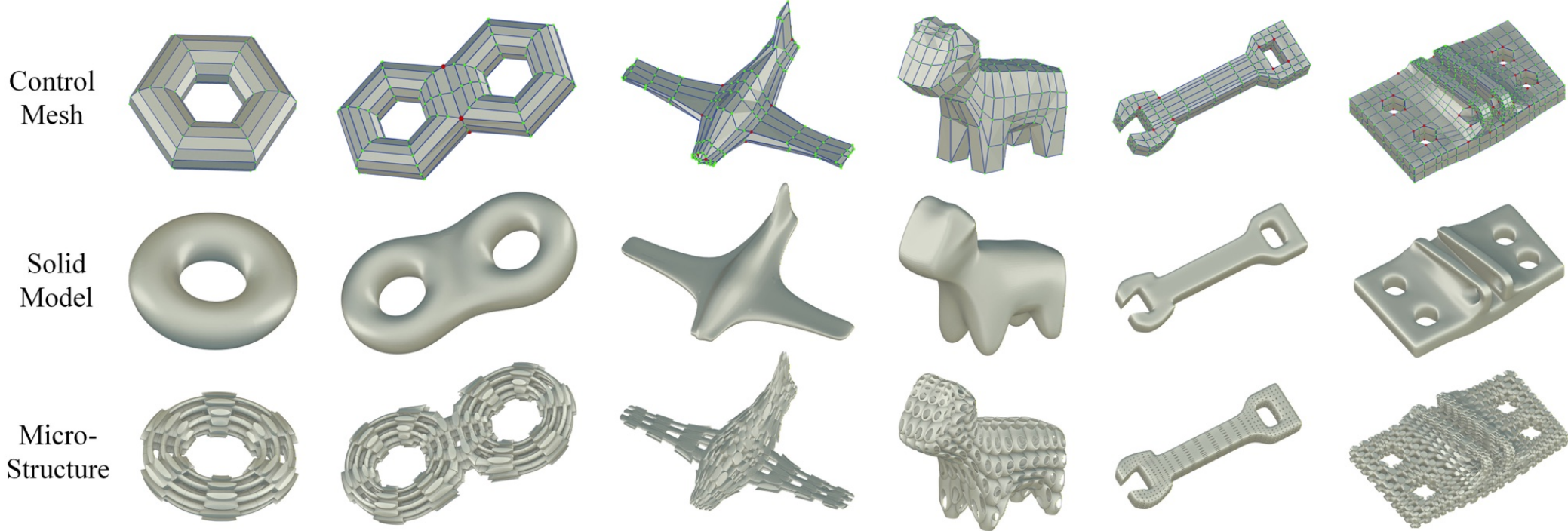}
\caption{The microstructures generated in different volumetric domains (with EVs marked in red in control meshes).}
\label{fig:topo_adapt}
\end{figure*}

\textbf{Cross-scale Boolean operation.}
Cross-scale composition requires both fields to be represented on the same parametric domain and in the same spline space.
Therefore, before combining a coarse-level primitive field $\phi_p^{\ell}(\mathbf{u})$ with a fine-level primitive field $\phi_q^{\ell+1}(\mathbf{u})$, we first lift the coarse primitive into the level-$(\ell+1)$ spline space using Eq.~\eqref{eq:ms-prolong}, obtaining $\phi_p^{\ell\rightarrow \ell+1}(\mathbf{u})$.

We then define the cross-scale composed primitive field (in the level-$(\ell+1)$ space) by the standard max/min rules:
\begin{equation}
\phi_{p\circ q}^{\ell+1}(\mathbf{u})=
\begin{cases}
\max\;\bigl(\phi_p^{\ell\rightarrow \ell+1}(\mathbf{u}),\, \phi_q^{\ell+1}(\mathbf{u})\bigr),
& \text{Union}\\
\min\;\bigl(\phi_p^{\ell\rightarrow \ell+1}(\mathbf{u}),\, -\phi_q^{\ell+1}(\mathbf{u})\bigr),
& \text{Difference}\\
\min\;\bigl(\phi_p^{\ell\rightarrow \ell+1}(\mathbf{u}),\, \phi_q^{\ell+1}(\mathbf{u})\bigr),
& \text{Intersection}
\end{cases}
\label{eq:ms-cross-bool-prim}
\end{equation}

By applying these operations recursively across consecutive refinement levels, multiscale patterns can be generated. Fig.~\ref{fig:multiscale_cell} illustrates three typical families of multiscale primitives—strut-based, shell-based, and solid-based—produced through these operations.

\subsubsection{On-demand Microstructure Evaluation}

On-demand microstructure evaluation is achieved by recursively traversing the multiscale tree structure. Given a query region \(\mathcal{Q}\) in the parametric domain and a target refinement level \(\ell^\ast\), we traverse the tree to collect the active leaves at the target level.
 
First, we identify the cells whose parametric domains intersect with the query region \(\mathcal{Q}\) and, when necessary, include neighboring cells to ensure that the spline basis functions that are nonzero inside \(\mathcal{Q}\) are fully covered. Starting from the coarsest level, we refine only the intersecting cells until reaching the target refinement level \(\ell^\ast\), producing a set of refined cells at this level. This ensures that only the necessary regions are refined, avoiding unnecessary global refinement. The set of refined cells is defined as:
\begin{equation}
    \mathcal{A} = \{ \mathcal{C}_i \mid \mathcal{C}_i \cap \mathcal{Q} \neq \emptyset, \ \ell(\mathcal{C}_i) = \ell^\ast \}
\end{equation}
where \(\mathcal{C}_i\) are the cells and \(\ell(\mathcal{C}_i)\) is their refinement level.

Within the refined cells at level \(\ell^\ast\), the microstructure primitives are instantiated in the corresponding spline space. Coarse primitives are rewritten in the \(\ell^\ast\) basis, while finer primitives are instantiated only inside the refined cells. Each node stores a node-local implicit field defined by its primitive expression. The implicit field is computed recursively by combining parent and child fields through Boolean operations (e.g., union, intersection, or difference):
\begin{equation}
    \phi_p^{\ell+1}(\mathbf{u}) = \text{Bool}\left( \phi_p^{\ell}(\mathbf{u}), \phi_c^{\ell+1}(\mathbf{u}) \right)
\end{equation}
where \(\text{Bool}\) denotes a Boolean operation, \(\phi_p^{\ell}(\mathbf{u})\) is the field at the parent node, and \(\phi_c^{\ell+1}(\mathbf{u})\) is the field at the child node.

This recursive process continues until the leaf nodes are reached. The final field on an active leaf is computed by recursively applying the Boolean combination along the root-to-leaf path. The overall field at a leaf \(\mathcal{L}\) is given by:
\begin{equation}
    \phi_L(\mathbf{u}) = \text{Bool}\left( \dots \text{Bool}\left( \phi_{\text{root}}(\mathbf{u}), \phi_{\ell_1}(\mathbf{u}) \right), \dots, \phi_{\ell_n}(\mathbf{u}) \right)
\end{equation}
where \(\phi_{\ell_1}(\mathbf{u}), \phi_{\ell_n}(\mathbf{u})\) are the fields from intermediate nodes, recursively combined until the leaf node \(\mathcal{L}\) is reached.

Once the implicit field is computed for the active cells at level \(\ell^\ast\), we extract the corresponding isosurface within the query region and then map it to the physical domain using the macro volumetric mapping.

\begin{figure}[t]
\centering
\includegraphics[width=0.4\textwidth]{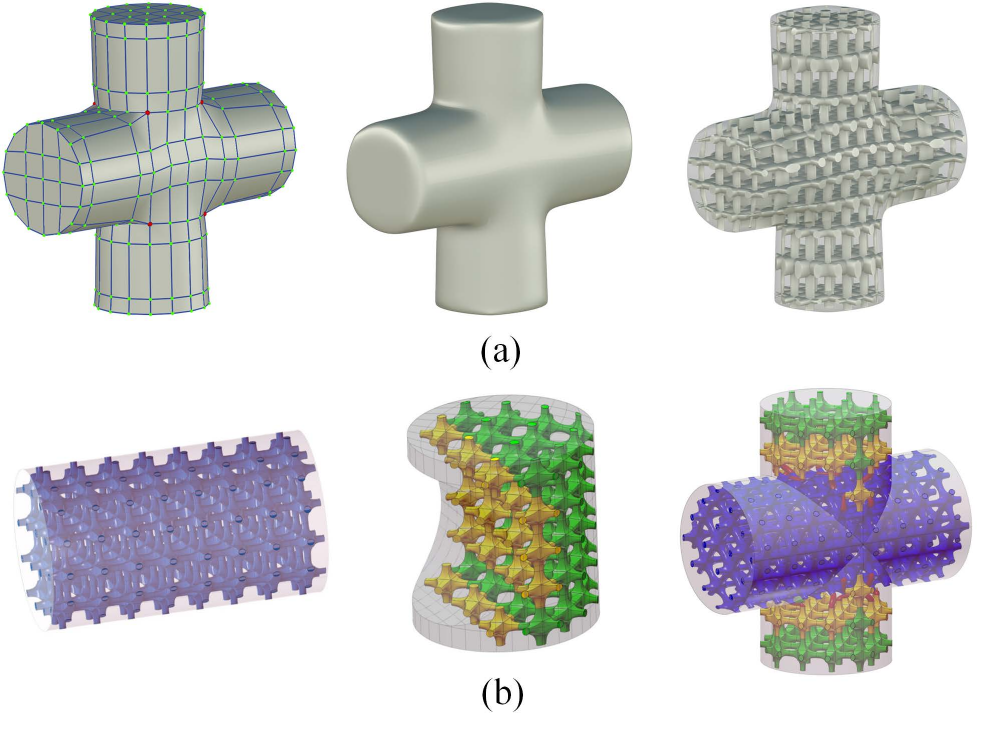}
\caption{The microstructures generated by our method (a) and by the method proposed in~\cite{2021_elber_conformal}(b) ((b) is referenced from~\cite{2021_elber_conformal}).}
\label{fig:topo_comparison}
\end{figure}

\section{Results and Discussion}
\label{sec:results}

The proposed method has been implemented in C++, running on an Intel Core i9-12900K CPU (5.20\, GHz) with 128\, GB RAM and Ubuntu 20.04 LTS operating system. To demonstrate the effectiveness of the method, we conducted five case studies, for which the visual results and quantitative statistics are reported. Specifically, Case Study 1 (Fig.~\ref{fig:topo_adapt} and Fig.~\ref{fig:topo_comparison}) considers several models spanning a wide range of shape complexity and topological configurations, demonstrating the broad applicability of our method for microstructure modeling. Case Study 2 (Fig.~\ref{fig:on-demand} and Table~\ref{tab:memory_cells}) and Case Study 3 (Fig.~\ref{fig:multiscale}) evaluate the proposed framework for on-demand generation and multiscale modeling, respectively. Case Study 4 (Fig.~\ref{fig:integrity} and Fig.~\ref{fig:integrity_compare}) examines the preservation of geometric integrity during interactive editing. Finally, Case Study 5 (Fig. 13–15) demonstrates the potential of the proposed method for advanced microstructure design, for example, generating gradient microstructures.

\subsection{Examples}

Case Study 1 evaluates the method’s applicability across diverse scenarios by considering a set of representative volumetric models with varying shape complexity and topology. These models cover a wide range of complexity, characterized by both the genus of the volumetric domain and the number of EVs in the control meshes, with each ranging from zero to multiple. The statistical data for these models are provided in Table~\ref{tab:model_statistics}, including the number of control vertices (abbreviated as CVs), hexahedral elements in the control meshes (abbreviated as HEs), EVs, EEs, and domain genus.
The control meshes, solid models, and microstructures of these models are shown in Fig.~\ref{fig:topo_adapt}.
In addition, we generate microstructures in a representative non-tensor-product volumetric domain (Fig.~\ref{fig:topo_comparison}(a)), namely the cross-shaped model used in~\cite{2021_elber_conformal} to enable a direct comparison with their method. The resulting geometry of our method is continuous, while ~\cite{2021_elber_conformal} still needs to generate an additional layer of cells to connect the separated parts of microstructures.

\begin{figure*}[t]
\centering
\includegraphics[width=0.6\textwidth]{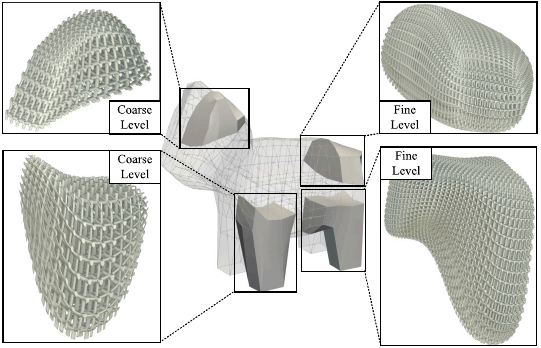}
\caption{On-demand microstructure generation in four local regions at different hierarchies.}
\label{fig:on-demand}
\end{figure*}

\begin{table*}[t]
  \centering
  \caption{Comparison of geometric complexity and storage cost.}
  \label{tab:memory_cells}
  \setlength{\tabcolsep}{12pt}
  
  \resizebox{0.8\textwidth}{!}{
  \begin{tabular*}{0.9\textwidth}{c c r r r r c}
\toprule
    \multirow{2}{*}{Models} & \multirow{2}{*}{Level} & \multirow{2}{*}{No. of Cells} & \multirow{2}{*}{Equiv. Triangles} & \multicolumn{2}{c}{Storage Cost} & Ratio \\
    \cmidrule(lr){5-6} 
     & & & & Ours & Explicit & (Expl./Ours) \\
    \midrule
    
    \multirow{5}{*}{Torus} 
      & 0 & 54        & $225,272$          & 65.33 KB  & 4.12 MB   & $63\times$ \\
      & 1 & 432       & $2,011,894$        & 586.18 KB & 47.93 MB  & $82\times$ \\
      & 2 & 3,456     & $16,807,219$       & 4.96 MB   & 471.20 MB & $95\times$ \\
      & 3 & 27,648    & $137,600,004$      & 41.37 MB  & 4.24 GB   & $105\times$ \\
      & \textbf{4} & \textbf{221,184} & \boldmath$1,100,791,185$\unboldmath & \textbf{337.97 MB} & \textbf{37.96 GB} & \textbf{112$\times$} \\
    \midrule
    
    \multirow{5}{*}{Cow} 
      & 0 & 240       & $1,035,800$        & 306.01 KB & 21.24 MB  & $69\times$ \\
      & 1 & 1,920     & $9,088,308$        & 2.61 MB   & 231.17 MB & $88\times$ \\
      & 2 & 15,360    & $76,592,885$       & 22.34 MB  & 2.14 GB   & $98\times$ \\
      & 3 & 122,880   & $628,050,001$      & 185.03 MB & 19.98 GB  & $110\times$ \\
     & \textbf{4} & \textbf{983,040} & \boldmath$5,068,265,226$\unboldmath & \textbf{1.47 GB} & \textbf{OOM} & \textbf{-} \\
    \midrule
    
    \multirow{5}{*}{Bracket} 
      & 0 & 604       & $2,163,096$        & 725.92 KB & 43.43 MB  & $60\times$ \\
      & 1 & 4,832     & $21,684,208$       & 6.34 MB   & 554.52 MB & $87\times$ \\
      & 2 & 38,656    & $191,076,221$      & 55.15 MB  & 5.38 GB   & $100\times$ \\
      & 3 & 309,248   & $1,561,251,516$    & 461.20 MB & 50.73 GB  & $113\times$ \\
      & \textbf{4} & \textbf{2,473,984} & \boldmath$12,578,966,700$\unboldmath & \textbf{3.69 GB}  & \textbf{OOM} & \textbf{-} \\
      
    \bottomrule
    \multicolumn{7}{l}{\footnotesize OOM: out-of-memory; the explicit method exceeds the 128\,GB memory limit on our system.}
  \end{tabular*}
  } 
\end{table*}

Case Study 2 evaluates the proposed framework for compact representation and on-demand generation. Our hierarchical representation allows for storing only a coarse mesh and refining local areas as needed. This method results in a much lower memory cost than trivially storing the fully instantiated microstructure geometry.
We validate this advantage by comparing memory consumption against a baseline that explicitly stores the complete final geometry (i.e., the extracted triangle mesh). Two scaling regimes are considered: increasing the number of coarse-level cells and increasing the refinement depth (which also increases the total number of cells). The resulting storage statistics are summarized in Table~\ref{tab:memory_cells}. We further present an on-demand generation example (Fig.~\ref{fig:on-demand}), in which the final geometry is instantiated only within user-specified regions, avoiding unnecessary computation in the rest of the domain and thereby enabling large-scale microstructure generation.

Case Study 3 showcases the proposed framework for multiscale microstructure modeling. Starting from a coarse microstructure, we refine user-specified regions to introduce fine-scale features while keeping the remainder of the domain unchanged. Refinement can be hierarchical: regions refined at a given level can be further refined, resulting in geometric details distributed across multiple scales. As shown in Fig.~\ref{fig:multiscale}, we present three-level microstructure examples using three representative cell types: strut-, shell-, and solid-based cells. Across all examples, locally refined regions exhibit increased geometric richness while preserving coherent transitions to the surrounding coarse structure.

For Case Study 4, we apply controlled macroscale shape modifications (e.g., deformation) while keeping the microstructure specification fixed, and generate both single-scale and multiscale microstructures on the original and edited volumetric models using identical microstructure parameters. Qualitative comparisons are shown in Fig.~\ref{fig:integrity}. In parallel, we perform the same experiment using the conventional design workflow built in the powerful commercial CAD software package, Siemens NX, as a baseline (Fig.~\ref{fig:integrity_compare}). In contrast to our method, this baseline does not maintain geometric integrity under macroscale changes and may introduce artifacts such as isolated fragments.

Case Study 5 explores the capabilities of the proposed framework for advanced microstructure design, with a focus on gradient structures. The process begins with a prescribed offset field, which is applied directly to the implicit field defining the microstructure. This offset field enables control over local geometric variations to introduce spatially varying features (see Fig.~\ref{fig:gradient_design}). Furthermore, by manipulating the offset field, the method is able to generate smoother transitions at scale interfaces, improving structural performance by reducing stress concentrations, especially in strut-based multiscale microstructures (see Fig.~\ref{fig:transition_compare}). To assess the performance of gradient designs, comparative displacement simulations are conducted between gradient and non-gradient microstructures under identical boundary conditions, with the gradient microstructures subjected to twice the load applied to the non-gradient ones. As shown in Fig.~\ref{fig:gradient_compare}, gradient microstructures exhibit improved deformation behavior. This case study highlights the potential of the framework to design microstructures with controlled gradients directly via implicit representations.

\begin{figure*}[t]
\centering
\includegraphics[width=0.85\textwidth]{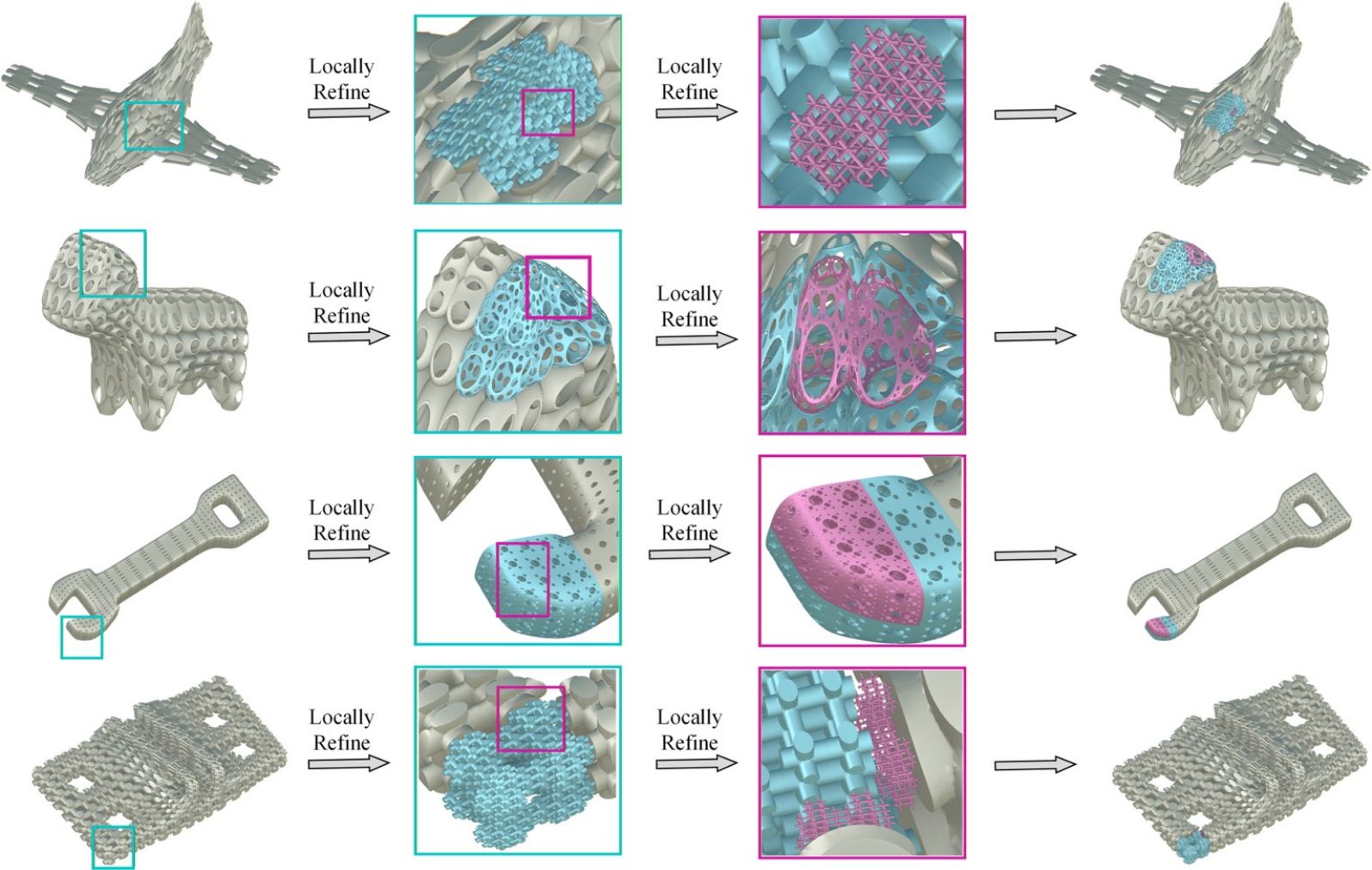}
\caption{The multiscality introduced into the microstructure through local refinement.}
\label{fig:multiscale}
\end{figure*}

\begin{figure*}[t]
\centering
\includegraphics[width=0.9\textwidth]{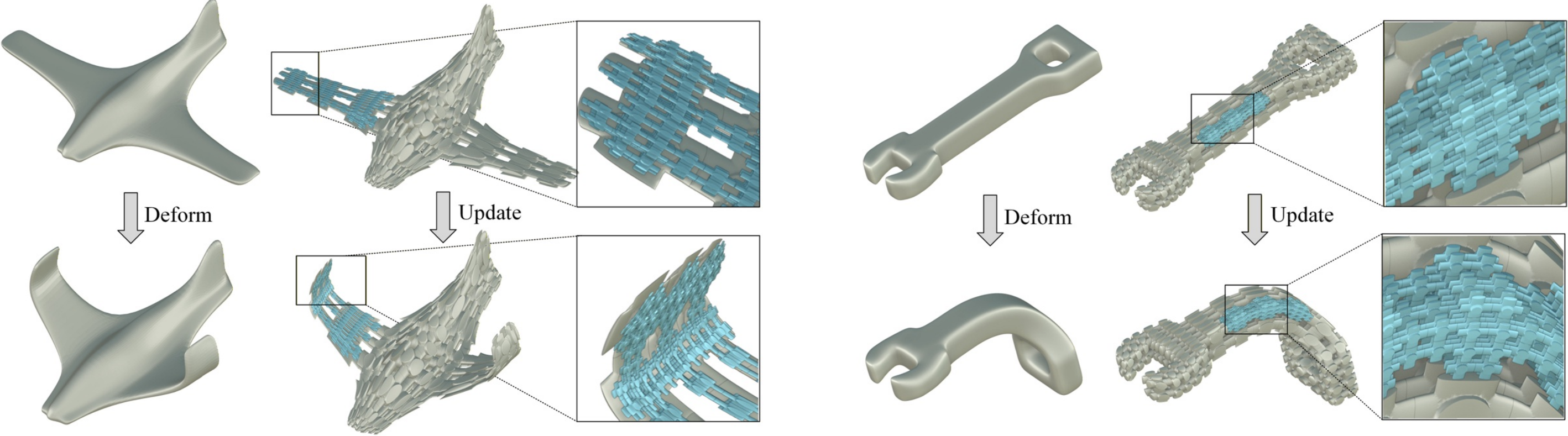}
\caption{Geometric association across multiple scales under macro-shape deformation.}
\label{fig:integrity}
\end{figure*}

\begin{figure}[t]
\centering
\includegraphics[width=0.5\textwidth]{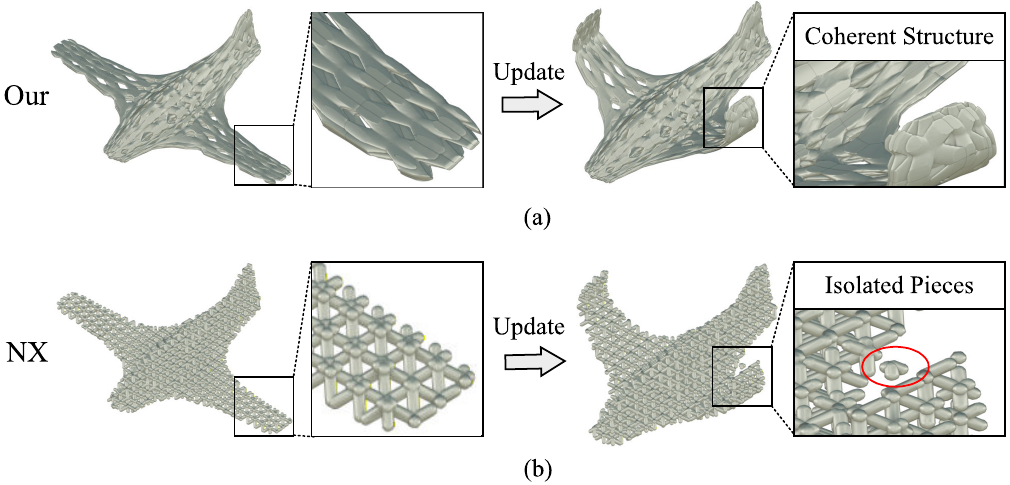}
\caption{Comparison of the geometric integrity of generated microstructures under deformation between our method (a) and Siemens NX (b).}
\label{fig:integrity_compare}
\end{figure}

\begin{table*}[b]
\centering
\caption{Computation time for modeling information generation and evaluation at different resolutions.}
\label{tab:efficiency_statistics}
\begin{tabular*}{0.75\textwidth}{c c c c c c c c c c c}
\toprule
 & \multicolumn{2}{c}{$2^4*2^4*2^4$} & \multicolumn{2}{c}{$2^5*2^5*2^5$} & \multicolumn{2}{c}{$2^6*2^6*2^6$} & \multicolumn{2}{c}{$2^7*2^7*2^7$} & \multicolumn{2}{c}{$2^8*2^8*2^8$} \\
 \cmidrule(lr){2-11}
 & Info & Eval & Info & Eval & Info& Eval& Info& Eval& Info& Eval\\
\midrule
 Regular&5ms &8ms &5ms &17ms &5ms &45ms &5ms & 200ms&6ms &1.3s \\

  Irregular&56ms &2.3s &59ms &2.5s &57ms &2.9s &57ms &4.7s &58ms &17.5s \\
\bottomrule
\end{tabular*}
\end{table*}

\begin{figure}[t]
\centering
\includegraphics[width=0.5\textwidth]{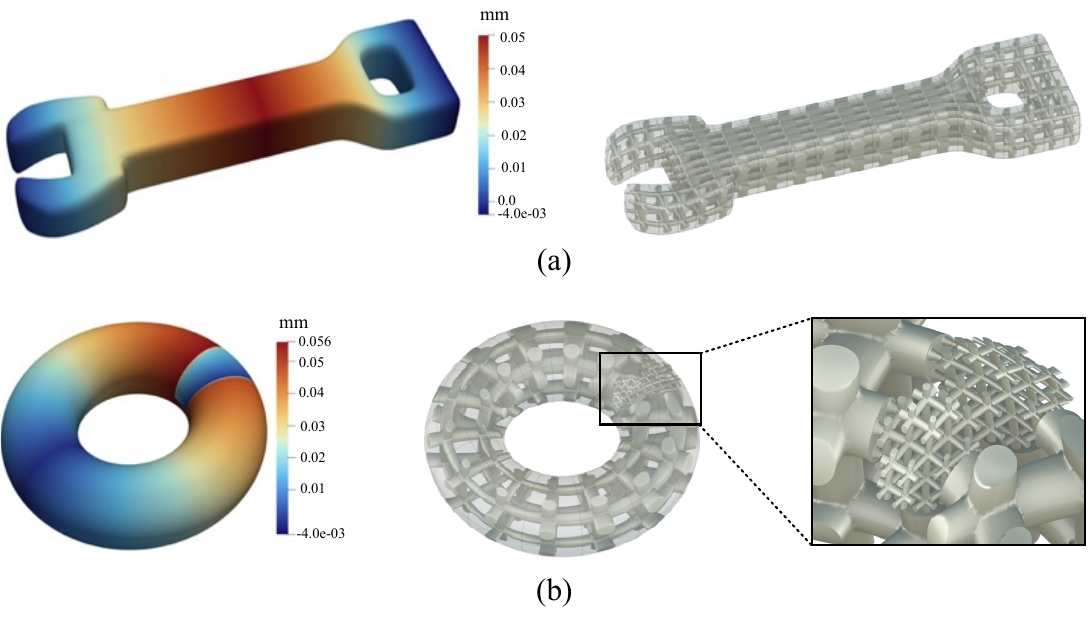}
\caption{Examples of single-scale (a) and multi-scale (b) gradient designs generated using the proposed framework with prescribed offset fields.}
\label{fig:gradient_design}
\end{figure}

\begin{figure}[t]
\centering
\includegraphics[width=0.48\textwidth]{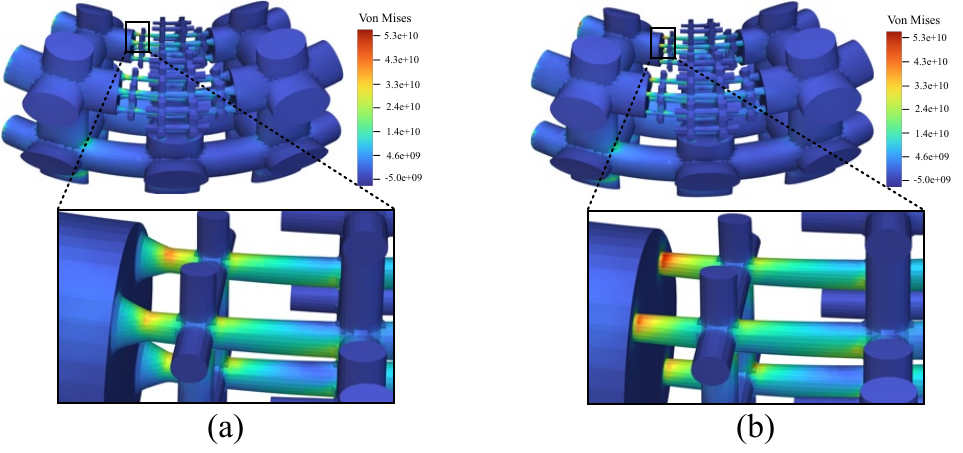}
\caption{Stress distribution at scale interfaces for designs with (a) and without (b) smooth transitions.}
\label{fig:transition_compare}
\end{figure}

\begin{figure*}[t]
\centering
\includegraphics[width=0.9\textwidth]{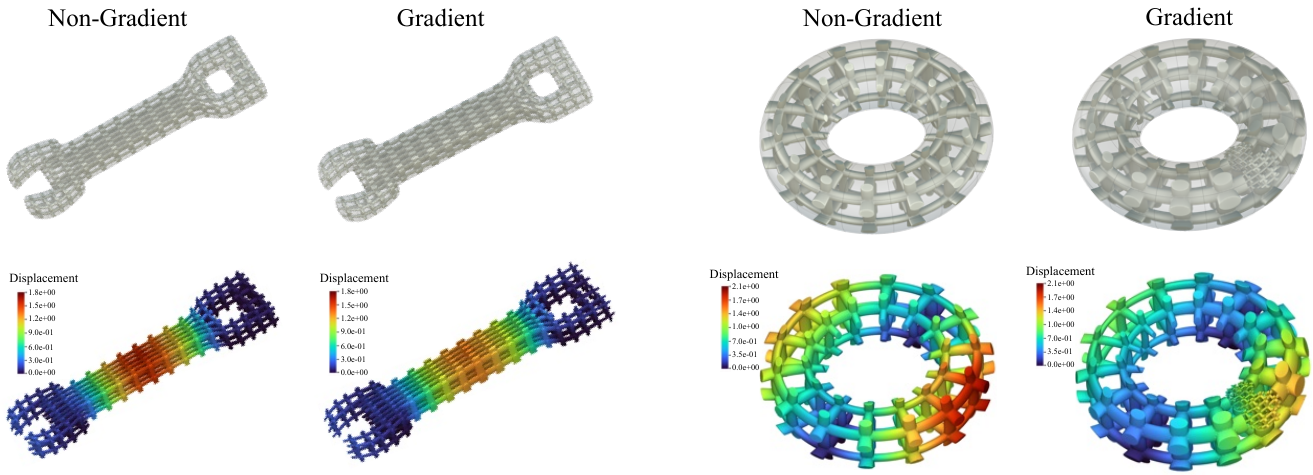}
\caption{Comparison of non-gradient and gradient microstructures under displacement simulation. Top row: geometry of the microstructures; bottom row: corresponding displacement fields.}
\label{fig:gradient_compare}
\end{figure*}

\subsection{Discussions and Limitations}
\label{sec:discussion}

The above examples and comparisons collectively demonstrate that the proposed framework provides an effective solution for the modeling of large-scale and multiscale microstructures. The results shown in Fig.~\ref{fig:topo_adapt} indicate that the proposed framework can represent complex volumetric shapes and generate corresponding microstructures on top of them, thereby enlarging the design space. The comparison in Fig.~\ref{fig:topo_comparison} further demonstrates that our method can effectively maintain the connectivity among adjacent cells in microstructures, which directly validates the advantage of the iso-parametric representation of the modeling framework.

From Fig.~\ref{fig:multiscale}, we can see that fine-scale geometric features can be introduced selectively in user-specified regions while leaving the rest of the domain unchanged. Refined regions can be further refined to achieve nested levels of detail, and the resulting multiscale microstructures exhibit coherent transitions between refined and unrefined areas. Additionally, by leveraging the flexibility of the subdivision rules, our method is capable of representing microstructures with sharp edges, as demonstrated by the zoom-in of the third row in Fig.~\ref{fig:multiscale}. Meanwhile, the statistics in Table~\ref{tab:memory_cells} highlight the representational compactness of the proposed framework, storing only the coarse control mesh together with a lightweight refinement specification and the associated coefficients (including the incremental value to change coefficients to edit geometry). Beyond storage, the on-demand generation example demonstrates that final geometry can be selectively instantiated for user-specified regions and refinement levels, avoiding unnecessary generation elsewhere. This ability makes the framework well-suited for handling large-scale microstructures in practice.

The results in Fig.~\ref{fig:integrity} show that the microstructure geometry can update following the macro-shape deformation. This demonstrates that our method provides a reliable mechanism to achieve geometric association across scales under interactive editing. In contrast, the explicit, geometry-centric modeling workflows in commercial CAD software packages, such as Siemens NX, may produce artifacts such as isolated microstructure fragments (see Fig.~\ref{fig:integrity_compare}). Furthermore, the results shown in Fig.~\ref{fig:gradient_design}-\ref{fig:gradient_compare} further demonstrate the capability of the proposed framework to serve as a flexible tool to achieve advanced microstructure design, e.g., gradient microstructures. Overall, these results fully validate the advances of our method for geometric modeling of large-scale and multiscale microstructures.

While the proposed framework is effective across the presented case studies, the evaluation cost near extraordinary vertices is generally higher than that on regular regions because recursive subdivision is involved in basis evaluation. This trend is consistent with the model-level observation that the Bracket model, which contains many extraordinary vertices and edges, evaluates more slowly than the Cow model, which has none. While irregular cases are consistently slower than regular ones but still remain within a practical range for most resolutions, as shown in Table~\ref{tab:efficiency_statistics}. For instance, at the commonly used resolution of $2^7*2^7*2^7$ per cell, corresponding to over 2,000,000 points to be evaluated, the evaluation of irregular cases takes approximately 4.7 seconds, whereas the time spent on modeling information generation (Info) remains negligible in comparison. Overall, further acceleration for highly irregular cases, such as caching and parallel/GPU-based evaluation, would be a worthwhile direction for future improvement.

Another consideration of the method is that the quality of the input hexahedral control mesh may affect the final modeling of microstructures. When the control mesh contains elements with extreme aspect ratios or noticeable size variations, the resulting parametric mapping may introduce distortions in certain local regions of the generated microstructures. This effect can be more evident near boundary layers, where the boundary-conformal requirement, together with local element irregularities, may lead to more distorted unit cells. Consequently, employing relatively uniform and well-shaped hexahedral control meshes generally leads to improved modeling fidelity.

\section{Conclusions}
\label{sec:conclusion}
A generative and iso-parametric modeling framework is proposed in this paper to support the geometric modeling of large-scale and multiscale microstructures. The framework effectively addresses the scalability and cross-scale geometric integrity challenges encountered in current CAD workflows. It resolves these issues by encoding microstructures as a coarse representation augmented with explicit refinement rules, so that fine details are generated on demand via localized refinement only within queried regions. Meanwhile, an iso-parametric representation is proposed, where the macro-shape and the microstructure are parameterized using the same spline basis family over a shared parametric domain. This unified representation, coupled with the refinement hierarchy, enables macro–meso–micro association of geometric details across scales and exact interface consistency between adjacent cells. A series of case studies and comparisons has also been conducted to validate the method.

While the proposed method has demonstrated effectiveness in the case studies, there is still room for improvement. Most notably, the evaluation efficiency of the current implementation could be further improved, as both the macro volumetric map and the microstructure field require a large number of ExVCC basis evaluations. This overhead becomes more noticeable when encountering extraordinary elements. Further improving evaluation performance, for example, through caching and parallel or GPU-based acceleration, is an important direction for future work.

Additionally, the current implementation includes only a subset of primitives and composition operators, which is sufficient primarily to demonstrate the feasibility of the proposed framework. Further expanding the primitive library, for example, by including complex TPMS-like primitives, and supporting richer field composition operators would significantly enhance the flexibility of multiscale microstructure modeling.

\section*{Acknowledgements}

This work has been funded by the ``Pioneer" and ``Leading Goose" R\&D Program of Zhejiang Province (No. 2024C01103), the National Natural Science Foundation of China (No. 62102355), and the Fundamental Research Funds for the Zhejiang Provincial Universities (No. K20241957, K20250142).

 \bibliographystyle{vancouver}
 \bibliography{cas-refs}

\appendix
\section{Primitive Field Computation}
\label{app:coeff}

This section details how we compute the spline-field representation of a primitive within a single cell. The main task is to determine the field parameters (coefficients) of each primitive in the local volumetric Catmull--Clark (VCC) function space. Our strategy proceeds in two steps:
(i) we convert an analytic tricubic primitive into an equivalent tricubic Bernstein (B\'ezier) form on the cell parametric domain, and
(ii) we map the resulting Bernstein coefficients to VCC coefficients via B\'ezier extraction matrix~\cite{2017_Wei_Bezier-extraction}.

\subsection{Bernstein coefficients for tricubic primitives}
On the cell parametric domain $(u,v,w)\in[0,1]^3$, we assume a typical primitive, e.g., cylinder, is specified by a tricubic polynomial:
\begin{equation}
f(u,v,w)=\sum_{i+j+k\le 3} a_{ijk}\,u^{i}v^{j}w^{k}
\label{eq:quad_poly}
\end{equation}
We rewrite $f$ in the tensor-product tricubic Bernstein basis:
\begin{equation}
f(u,v,w)=\sum_{p=0}^{3}\sum_{q=0}^{3}\sum_{r=0}^{3} C_{pqr}\,B_p^3(u)\,B_q^3(v)\,B_r^3(w)
=\mathbf{C}_e^{T}\mathbf{b}_e(u,v,w)
\label{eq:tri_bernstein}
\end{equation}
where $B_p^3(u)=\binom{3}{p}u^p(1-u)^{3-p}$, $\mathbf{b}_e$ stacks the $64$ tricubic Bernstein basis functions, and $\mathbf{C}_e\in\mathbb{R}^{64}$ stacks the coefficients $\{C_{pqr}\}$.

Specifically, each univariate monomial $u^i$ ($i\le 3$) admits a representation in the cubic Bernstein basis:
\begin{equation}
u^{i}=\sum_{p=0}^{3}\alpha^{(i)}_{p}\,B_p^3(u), \qquad i=0,1,2,3
\label{eq:mono_bernstein}
\end{equation}
where $\alpha^{(i)}_{p}$ are fixed constants, and analogously for $v^j$ and $w^k$. Substituting Eq.~\eqref{eq:mono_bernstein} into Eq.~\eqref{eq:quad_poly} and collecting terms in $B_p^3(u)B_q^3(v)B_r^3(w)$ yields:
\begin{equation}
C_{pqr}=\sum_{i+j+k\le 3} a_{ijk}\,\alpha^{(i)}_{p}\,\alpha^{(j)}_{q}\,\alpha^{(k)}_{r}
\label{eq:Cpqr_from_aijk}
\end{equation}
Equations~\eqref{eq:mono_bernstein}--\eqref{eq:Cpqr_from_aijk} therefore provide a direct analytic mapping from $\{a_{ijk}\}$ to $\mathbf{C}_e$.

\subsection{B\'ezier extraction-based coefficient conversion to VCC}
We next convert the tricubic Bernstein representation into a local ``VCC basis functions + coefficients'' form. Following the B\'ezier extraction framework~\cite{2017_Wei_Bezier-extraction}, we use an element-wise extraction matrix $\mathbf{M}$ to relate the local VCC basis functions and the tricubic Bernstein basis on the reference parametric cell:
\begin{equation}
\mathbf{B}(u,v,w)\approx \mathbf{M}^{T}\mathbf{b}_e(u,v,w)
\label{eq:bezier_extraction_approx}
\end{equation}
For regular elements, the relation becomes exact; for elements incident to extraordinary configurations, we still employ the same extraction matrix as a practical approximation.

Given the Bernstein form
\begin{equation}
f(u,v,w)=\mathbf{C}_e^{T}\mathbf{b}_e(u,v,w)
\label{eq:primitive_bernstein}
\end{equation}
we seek a VCC representation
\begin{equation}
f(u,v,w)=\mathbf{C}_{B}^{T}\mathbf{B}(u,v,w)
\label{eq:primitive_ccc}
\end{equation}
Substituting Eq.~\eqref{eq:bezier_extraction_approx} into Eq.~\eqref{eq:primitive_ccc} leads to an (approximate) coefficient relation
\begin{equation}
\mathbf{M}\mathbf{C}_{B}\approx \mathbf{C}_e
\label{eq:coef_conversion_approx}
\end{equation}
We compute $\mathbf{C}_{B}$ via a least-squares solve
\begin{equation}
\mathbf{C}_{B}=\left(\mathbf{M}^{T}\mathbf{M}+\lambda \mathbf{I}\right)^{-1}\mathbf{M}^{T}\mathbf{C}_e
\label{eq:coef_conversion_tikh}
\end{equation}
where $\lambda>0$ is a small regularization parameter for numerical stability.

For regular elements, Eq.~\eqref{eq:coef_conversion_tikh} recovers an accurate cell-wise VCC representation of the primitive. For extraordinary elements, the extraction relation in Eq.~\eqref{eq:bezier_extraction_approx} is used as an approximation; consequently, the converted coefficients may slightly deviate from an ideal analytic primitive shape. In practice, this small deviation has a negligible impact on the microstructure modeling, as the examples in this study demonstrate. If desired, post-processing can be applied to enforce the intended sharing of boundary degrees of freedom across adjacent unit cells, further enhancing geometric consistency.





\end{document}